\PassOptionsToPackage{force}{filehook} 
\documentclass[runningheads,a4paper]{llncs}
\pdfoutput=1
\usepackage{graphicx}
\usepackage{subfigure}
\usepackage[table, dvipsnames]{xcolor}
\usepackage{amsmath,amssymb}
\usepackage{mathtools}
\usepackage{wrapfig}
\usepackage{longtable}
\usepackage{booktabs}
\usepackage{pythonhighlight}
\usepackage{array}
\usepackage{pgffor}
\usepackage{fontawesome}

\spnewtheorem{assumption}{Assumption}{\bfseries}{\itshape}

\makeatletter
\def\bstr{b}
\def\bfstr{bf}
\def\cstr{c}
\def\fstr{f}
\def\strLst{A,B,C,D,d,E,F,G,H,I,J,K,L,M,N,O,P,Q,R,S,T,U,V,W,X,Y,Z}

\newcommand{\MkB}[1]{\expandafter\def\csname\bstr#1\endcsname{\mathbb{#1}}}
  \@for\i:=\strLst\do{%
    \expandafter\MkB \i     }

\newcommand{\MkBF}[1]{\expandafter\def\csname\bfstr#1\endcsname{\mathbf{#1}}}
  \@for\i:=\strLst\do{%
    \expandafter\MkBF \i     }

\newcommand{\MkCal}[1]{\expandafter\def\csname\cstr#1\endcsname{\mathcal{#1}}}
  \@for\i:=\strLst\do{%
    \expandafter\MkCal \i     }

\newcommand{\MkFrak}[1]{\expandafter\def\csname\fstr#1\endcsname{\mathfrak{#1}}}
  \@for\i:=\strLst\do{%
    \expandafter\MkFrak \i     }
    
\makeatother

\newcommand{\MatchGT}[3]{\mathsf{M}^{{\text{\tiny $#1$}}}_{#2}(#3)}

\newcommand{\RMatchGT}[3]{\mathcal{M}^{{\text{\tiny $#1$}}}_{#2}(#3)}

\newcommand{\pO}[1]{\mathsf{PO}(#1)}

\newcommand{\mono}[1]{\mathsf{mono}(#1)}

\newcommand{\regmono}[1]{\mathsf{rm}(#1)}

\newcommand{\epi}[1]{\mathsf{epi}(#1)}

\newcommand{\mor}[1]{\mathsf{mor}(#1)}

\newcommand{\iso}[1]{\mathsf{iso}(#1)}

\newcommand{\obj}[1]{\mathsf{obj}(#1)}

\newcommand{\mIO}{\mathop{\varnothing}}

\newcommand{\comp}[3]{#1 \stackrel{#2}{\blacktriangleleft} #3}

\newcommand{\sqComp}[3]{#1 \stackrel{#2}{\sphericalangle} #3}

\newcommand{\Msum}[2]{\text{$\sum\nolimits_{\cM}$}(#1,#2)}

\newcommand{\mPOC}[2]{\cP(#1,#2)}

\newcommand{\FPA}[2]{\mathsf{FPA}(#1,#2)}

\newcommand{\git}{\mathbin{
  \mathchoice{/\mkern-6mu/}%
    {/\mkern-6mu/}%
    {/\mkern-5mu/}%
    {/\mkern-5mu/}}}%

\newcommand{\ti}[1]{%
 \ensuremath{\vcenter{\hbox{\includegraphics{diagrams/#1.pdf}}}}%
}

\colorlet{h1color}{blue!70!black} %
\colorlet{h2color}{orange!90!black} %
\colorlet{h3color}{blue!40!white} %
\colorlet{h4color}{green!40!black} %

\RequirePackage{hyperref}
\hypersetup{%
    pdfstartview=Fit,
    breaklinks=true,
    colorlinks=true
    }

\usepackage[strings]{underscore}

\usepackage[all]{hypcap}

\usepackage[scope,electronic]{knowledge} %

\knowledgeconfigure{notion}

\knowledge{anchor point}[anchor points|Anchor points]{notion}
\knowledge{diagnose file}{notion}
\knowledge{knowledge package}{notion,typewriter}
\knowledge{paper mode}{notion}
\knowledge{electronic mode}{notion}
\knowledge{composition mode}{notion}

\knowledge{knowledge}{notion}
\knowledge{knowledges}{synonym}
\knowledge{Knowledges}{synonym}

\knowledge{LaTeX}[latex|LATEX|Latex]{url={https://fr.wikipedia.org/wiki/LaTeX},  text=\LaTeX}

\knowledge{pdflatex}{scope={section:example},url={https://en.wikipedia.org/wiki/PdfTeX}, typewriter}
\begin{scope}\label{package}
   \knowledge{knowledge}{link=knowledge package,link scope}
\end{scope}

\knowledge{notion}
 | stable system of monics

\knowledge{notion}
 | stable under pullbacks

\knowledge{notion}
 | $\cM $-partial map classifier

\knowledge{notion}
 | regular mono (rm)-quasiadhesive category
 | rm-quasi-adhesive

\knowledge{notion}
 | multi-sum

\knowledge{notion}
 | $\cM $-adhesive

\knowledge{notion}
 | Pushouts along $\cM $-morphisms are stable under pullbacks
 | pushouts along $\cM $-morphisms are stable under pullbacks

\knowledge{notion}
 | Algorithm to compute $\mPOC {f}{m}$

\knowledge{notion}
 | finitary

\knowledge{notion}
 | universal property of multi-sums

\knowledge{notion}
 | exists a final pullback-complement (FPC)

\knowledge{notion}
 | epi-$\cM $-factorization

\knowledge{notion}
 | category $\mathbf {FinGraph}$

\knowledge{notion}
 | category $\mathbf {FinSGraph}$

\knowledge{notion}
 | quasi-topos

\knowledge{notion}
 | $\cM $-multi-pushout complement
 | $\cM $-multi-POC

\knowledge{notion}
 | Universal property of $\cP (f,m)$

\knowledge{notion}
 | Algorithm to compute $\cP (f,m)$

\knowledge{notion}
 | Pushouts along $\cM $-morphisms are pullbacks
 | pushouts along $\cM $-morphisms are pullbacks

\knowledge{notion}
 | Stability of $\cM $-morphisms under pushouts
 | stability of $\cM $-morphisms under pushouts
 | under pushouts

\knowledge{notion}
 | Stability of $\cM $-morphisms under pullbacks
 | stability of $\cM $-morphisms under pullbacks

\knowledge{notion}
 | Decomposition property of $\cM $-morphisms
 | decomposition property of $\cM $-morphisms

\knowledge{notion}
 | Pushout-pushout-(de-)composition
 | pushout-pushout composition
 | pushout composition
 | pushout-pushout decomposition

\knowledge{notion}
 | Pullback-pullback-(de-)composition
 | pullback-pullback decomposition

\knowledge{notion}
 | Pushout-pullback-decomposition
 | pushout-pullback decomposition

\knowledge{notion}
 | Pullback-pushout-decomposition
 | pullback-pushout decomposition

\knowledge{notion}
 | $\cM $-FPC augmentation

\knowledge{notion}
 | Horizontal FPC-FPC-(de-)composition
 | horizontal FPC composition
 | horizontal FPC decomposition
 
\knowledge{notion}
 | Vertical FPC-FPC-(de-)composition
 | vertical FPC composition

\knowledge{notion}
 | Vertical FPC-pullback decomposition
 | vertical FPC-pullback decomposition

\knowledge{notion}
 | universal property of pullbacks
 | Universal property of pullbacks (PBs)

\knowledge{notion}
 | Universal property of pushouts (POs)
 | universal property of pushouts

\knowledge{notion}
 | Universal property of final pullback complements (FPCs)
 | universal property of FPCs
 | final pullback complement (FPC)

\knowledge{notion}
 | pushouts along regular monos are also pullbacks

\knowledge{notion}
 | regular monos are stable under pushout

\knowledge{notion}
 | unions of regular subobjects are effective
 | effectiveness of binary unions of regular subobjects

\knowledge{notion}
 | epi-regular mono factorization

\knowledge{notion}
 | extremal monomorphisms
 
\knowledge{notion}
 | category $\mathbf {Graph}$

\knowledge{notion}
 | category $\mathbf {SGraph}$

\knowledge{notion}
 | strict initial object

\knowledge{notion}
 | $\cM $-initial

\knowledge{notion}
 | disjoint coproducts

\knowledge{notion}
 | solid quasi-topos

\knowledge{notion}
 | pushout $(*)$ in $\bfC $ is said to be \emph {stable under pullbacks}

\knowledge{notion}
 | final pullback complement (FPC) $(*)$ in $\bfC $ is said to be \emph {stable under pullbacks}

\knowledge{notion}
 | Every quasi-topos is \emph {rm-quasiadhesive}
 | rm-quasiadhesive

\knowledge{notion}
 | stability of FPCs under pullbacks
 | FPCs are stable under pullbacks

\knowledge{notion}
 | pushouts along regular monomorphisms are stable under pullbacks
 | stability of pushouts along regular monomorphisms under pullbacks

\knowledge{notion}
 | rm-adhesive category

\knowledge{notion}
 | pushouts along regular monomorphisms

\knowledge{notion}
 | rm-pushouts

\knowledge{notion}
 | pushouts along regular monos are pullbacks

\knowledge{notion}
 | Universal property of $\mPOC {f}{m}$
 | universal property of multi-POCs

\knowledge{notion}
 | $\cM $-FPC augmentation (FPA)
 | universal property of FPAs

\knowledge{notion}
 | class of FPAs

\knowledge{notion}
 | stability of stable pushouts under pullbacks

\knowledge{notion}
 | General SqPO-rewriting semantics

\knowledge{notion}
 | set of SqPO-admissible matches

\knowledge{notion}
 | SqPO-type direct derivation

\knowledge{notion}
 | set of SqPO-type admissible matches of rules

\knowledge{notion}
 | SqPO-type rule composition

\knowledge{notion}
 | General DPO-rewriting semantics

\knowledge{notion}
 | set of DPO-type admissible matches of rules

\knowledge{notion}
 | set of DPO-admissible matches

\knowledge{notion}
 | DPO-type direct derivation

\knowledge{notion}
 | DPO-type rule composition

\begin{document}
\title{Concurrency Theorems for \\Non-linear Rewriting Theories\thanks{This is an extended version (containing additional technical appendices) of a paper with the same tittle accepted for \href{https://icgt2021.di.unipi.it}{ICGT 2021}.}}
\titlerunning{Concurrency Theorems for Non-linear Rewriting Theories}
\author{Nicolas Behr\inst{1}\,\textsuperscript{\faEnvelopeO}
\and Russ Harmer\inst{2}
\and Jean Krivine\inst{1}}
\authorrunning{N. Behr, R. Harmer and J. Krivine}
\institute{Universit\'{e} de Paris, CNRS, IRIF\\
8 Place Aur\'{e}lie Nemours, 75205 Paris Cedex 13, France\\
\email{nicolas.behr@irif.fr}
\email{jean.krivine@irif.fr}
\and
Universit\'{e} de Lyon, ENS de Lyon, UCBL, CNRS, LIP\\
46 allée d'Italie, 69364 Lyon Cedex 07, France\\
\email{russell.harmer@ens-lyon.fr}
}
\maketitle
\begin{abstract}
Sesqui-pushout (SqPO) rewriting along non-linear rules and for monic matches is well-known to permit the modeling of fusing and cloning of vertices and edges, yet to date, no construction of a suitable concurrency theorem was available. The lack of such a theorem, in turn, rendered compositional reasoning for such rewriting systems largely infeasible. We develop in this paper a suitable concurrency theorem for non-linear SqPO-rewriting in categories that are quasi-topoi (subsuming the example of adhesive categories) and with matches required to be regular monomorphisms of the given category. Our construction reveals an interesting ``backpropagation effect'' in computing rule compositions. We derive in addition a concurrency theorem for non-linear double pushout (DPO) rewriting in rm-adhesive categories. Our results open non-linear SqPO and DPO semantics to the rich static analysis techniques available from concurrency, rule algebra and tracelet theory.
\end{abstract}

\section{Introduction}

Sesqui-pushout (SqPO) graph transformation was introduced~\cite{Corradini_2006} as an extension of single-pushout rewriting that accommodates the possibility of non-input-linear\footnote{In this paper, we follow the conventions of compositional rewriting theory~\cite{behrRaSiR}, i.e., we speak of ``input''/``output'' motifs of rules, as opposed to ``left''/``right'' motifs in the traditional literature~\cite{ehrig:2006fund}.}  rules. The result of such a rewrite is specified abstractly by the notion of \emph{final pullback complement (FPC)}~\cite{dyckhoff1987exponentiable}, a categorical generalization of the notion of set difference: the FPC of two composable arrows, $f : A \rightarrow B$ and $g : B \rightarrow D$ is the largest, \emph{i.e.} least general, $C$ together with arrows $g' : A \rightarrow C$ and $f' : C \rightarrow D$ for which the resulting square is a pullback (PB). The extension of graph transformation to input-non-linear rules allows for the expression of the natural operation of the \emph{cloning} of a node, or an edge (when the latter is meaningful), as explained in~\cite{Corradini_2006,Corradini_2015,BRAATZ2011246} . More recently, such rules have also been used to express operations such as concept refinement in schemata for graph databases~\cite{bonifati2019schema} and, more generally, in graph-based knowledge representation~\cite{harmer2020knowledge}. In combination with output-non-linear rules, as for (non-linear) double- or single-pushout rewriting, SqPO thus allows the expression of all the natural primitive operations on graphs: addition and deletion of nodes and edges; and cloning and merging of nodes and edges.

In this paper, we study the categorical structure required in order to support SqPO rewriting and establish that \emph{quasi-topoi}~\cite{adamek2006,COCKETT2002223,COCKETT200361,quasi-topos-2007,garner2012axioms} naturally possess all the necessary structure to express the effect of SqPO rewriting and to prove the concurrency theorem for fully general non-linear rules. This significantly generalizes previous results on concurrency theorems for linear SqPO-rewriting over adhesive categories~\cite{nbSqPO2019} and for linear SqPO-rewriting for linear rules with conditions in $\cM$-adhesive categories~\cite{behrRaSiR,BK2020}. In terms of SqPO-rewriting for generic rules, previous results were rather sparse and include work on polymorphic SqPO-rewriting~\cite{Loewe_2015} and on reversible SqPO rewriting~\cite{reversibleSqPO,harmer2020reversibility}, where~\cite{harmer2020reversibility} in particular introduced a synthesis (but not an analysis) construction for reversible non-linear SqPO rules without application conditions which motivated the present paper. 

An interesting technical aspect of basing our constructions on quasi-topoi concerns the rewriting of simple directed graphs, which constitutes one of the running examples in this paper: unlike the category of directed multigraphs (which constitutes one of the prototypical examples of an adhesive category~\cite{ls2004adhesive}), the category of simple graphs is neither adhesive nor quasi-adhesive~\cite{quasi-topos-2007}, but it is in fact only a quasi-topos~\cite{adamek2006,quasi-topos-2007}, and as such also an example of an rm-quasi-adhesive~\cite{garner2012axioms} and of an $\cM$-adhesive category~\cite{ehrig2004adhesive,ehrig2010categorical,10.1007/978-3-642-15928-2_17,ehrig2014mathcal}. 

Our proof of the concurrency theorem relies on the existence of certain structures in quasi-topoi that, to the best of our knowledge, have not been previously noted in the literature (cf.\ Section~\ref{sec:special}): restricted notions of \emph{multi-sum} and \emph{multi-pushout complement (mPOC)}, along the lines of the general theory of multi-(co-) limits due to Diers~\cite{diers1978familles}, and a notion of \emph{FPC-pushout-augmentation (FPA)}. The notion of multi-sum provides a generalization of the property of effective unions (in adhesive categories) that guarantees that all necessary monos are regular. The notions of mPOC and FPA handle the ``backward non-determinism'' introduced by non-linear rules: given a rule and a matching from its output motif, we cannot---unlike with linear or reversible non-linear rules---uniquely determine a matching from the input motif of the rule.\\

\noindent\textbf{Related work }Conditions under which FPCs are guaranteed to exist have been studied in \cite{dyckhoff1987exponentiable}, and more concretely and of particular relevance to our approach in~\cite{Corradini_2015}, which provides a direct construction assuming the existence of appropriate partial map classifiers~\cite{10.1007/978-3-642-15928-2_17,COCKETT200361}. We make additional use of these partial map classifiers in order to construct mPOCs in a quasi-topos (Section~\ref{sec:special}). Our construction is a mild, but necessary for our purposes, generalization of the notion of minimal pushout complement defined in~\cite{BRAATZ2011246} that requires the universal property with respect to a larger class of encompassing pushouts (POs)---precisely analogous to the definition of FPC. However, there is the additional complexity that, for our purposes, PO complements are not uniquely determined, and we must therefore specify a family of solutions that collectively satisfy this universal property (\`a la Diers \cite{diers1978familles}). We also exploit the epi-regular mono factorization~\cite{adamek2006} in quasi-topoi in order to construct multi-sums---with respect to co-spans of regular monos---and FPAs. Our overall approach relates closely to the work of Garner and Lack on rm-quasi-adhesive categories~\cite{garner2012axioms}, which provide an abstract setting for graph transformation that accommodates the technical particularities of simple graphs---notably the fact that the `exactness' direction of the van Kampen condition fails in general for cubes where the vertical arrows, between the two PO faces, are not regular.

\section{Quasi-topoi}\label{sec:qt}

In this section, we will demonstrate that quasi-topoi form a natural setting within which non-linear sesqui-pushout (SqPO) rewriting is well-posed. Quasi-topoi have been considered in the context of rewriting theories as a natural generalization of adhesive categories in~\cite{lack2005adhesive}. While several adhesive categories of interest to rewriting are topoi, including in particular the category $\mathbf{Graph}$ of directed multigraphs (cf.\ Definition~\ref{def:Graph}), it is not difficult to find examples of categories equally relevant to rewriting theory that fail to be topoi. A notable such example is the category $\mathbf{SGraph}$ of directed simple graphs (cf.\ Definition~\ref{def:SGraph}).

We will demonstrate that quasi-topoi combine all technical properties necessary such as to admit the construction of non-linear sesqui-pushout semantics over them. We will first list these abstract properties, and illustrate them via the two aforementioned paradigmatic examples of topoi and quasi-topoi.

Let us first recall a number of results from the work of Cockett and Lack~\cite{COCKETT2002223,COCKETT200361} on restriction categories. We will only need a very small fragment of their theory, namely the definition and existence guarantees for $\cM$-partial map classifiers, so we will follow mostly \cite{Corradini_2015}. We will in particular not be concerned with the notion of $\cM$-partial maps itself.

\begin{definition}[\cite{COCKETT2002223}, Sec.~3.1]
For a category $\bfC$, \AP a \emph{\intro{stable system of monics}} $\cM$ is a class of monomorphisms of $\bfC$ that $(i)$ includes all isomorphisms, $(ii)$ is stable under composition, and $(iii)$ is \AP\intro{stable under pullbacks} (i.e., if $(f',m')$ is a pullback of $(m,f)$ with $m\in \cM$, then $m'\in \cM$). Throughout this paper, we will reserve the notation $\rightarrowtail$ for monics in $\cM$, and $\hookrightarrow$ for generic monics.
\end{definition}

\begin{definition}[\cite{Corradini_2015}, Sec.~2.1; compare~\cite{COCKETT200361}, Sec.~2.1]
For a \kl{stable system of monics} $\cM$ in a category $\bfC$, an \AP\emph{\intro{$\cM$-partial map classifier}} $(T,\eta)$ is a functor $T:\bfC\rightarrow \bfC$ and a natural transformation $\eta:ID_{\bfC}\xrightarrow{.} T$ such that 
\begin{enumerate}
\item for all $X\in \obj{\bfC}$, $\eta_X:X\rightarrow T(X)$ is in $\cM$
\item for each span $(A\xleftarrow{m}X\xrightarrow{f}B)$ with $m\in \cM$, there exists a unique morphism $A\xrightarrow{\varphi(m,f)}T(B)$ such that $(m,f)$ is a pullback of $(\varphi(m,f),\eta_B)$.
\end{enumerate}
\end{definition}

\begin{proposition}[\cite{Corradini_2015},  Prop.~6]
For every $\cM$-partial map classifier $(T,\eta)$, $T$ preserves pullbacks, and $\eta$ is Cartesian, i.e., for each $X\xrightarrow{f}Y$, $(\eta_x,f)$ is a pullback of $(T(f),\eta_Y)$.
\end{proposition}

\begin{definition}[\cite{quasi-topos-2007}, Def.~9]
A category $\bfC$ is a \AP\emph{\intro{quasi-topos}} iff
\begin{enumerate}
\item it has finite limits and colimits
\item it is locally Cartesian closed
\item it has a regular-subobject-classifier.
\end{enumerate}
\end{definition}

Based upon a variety of different results from the rich literature on quasi-topoi, we will now exhibit that quasi-topoi indeed possess all technical properties required in order for non-linear SqPO-rewriting to be well-posed:
\begin{corollary}\label{cor:main}
Every \emph{quasi-topos} $\bfC$ enjoys the following properties:
\begin{itemize}
\item It has (by definition) a \kl{stable system of monics} $\cM=\regmono{\bfC}$ (the class of regular monos), which coincides with the class of \AP\emph{\intro{extremal monomorphisms}}~\cite[Cor.~28.6]{adamek2006}, i.e., if $m=f\circ e$ for $m\in \regmono{\bfC}$ and $e\in \epi{\bfC}$, then $e\in \iso{\bfC}$.
\item It has (by definition) a \kl{$\cM$-partial map classifier} $(T,\eta)$.
\item It is \AP\emph{\intro{rm-quasi-adhesive}},  i.e., it has \AP\intro{pushouts along regular monomorphisms}, these are \AP\emph{\intro[rm-pushouts]{stable under pullbacks}}, and \AP\intro{pushouts along regular monos are pullbacks}~\cite{garner2012axioms}.
\item It is \AP\emph{\intro{$\cM$-adhesive}}~\cite[Lem.~13]{10.1007/978-3-642-15928-2_17}.
\item For all pairs of composable morphisms $A\xrightarrow{f}B$ and $B\xrightarrow{m}C$ with $m\in \cM$, there \AP\emph{\intro{exists a final pullback-complement (FPC)}} $A\xrightarrow{n}F\xrightarrow{g}C$, and with $n\in \cM$ (\cite[Thm.~1]{Corradini_2015}; cf.\ Theorem~\ref{thm:FPC}).
\item It possesses an \AP\emph{\intro{epi-$\cM$-factorization}}~\cite[Prob.~28.10]{adamek2006}: each morphism $A\xrightarrow{f}B$ factors as $f=m\circ e$, with morphisms $A\xrightarrow{e}\bar{B}$ in $\epi{\bfC}$ and $\bar{B}\xrightarrow{m}A$ in $\cM$ (uniquely up to isomorphism in $\bar{B}$).
\item It possesses a \AP\emph{\intro{strict initial object}} $\mIO\in \obj{\bfC}$~\cite[A1.4]{johnstone2002sketches}, i.e., for every object $X\in \obj{\bfC}$, there exists a morphism $i_X:\mIO\rightarrow X$, and if there exists a morphism $X\rightarrow \mIO$, then $X\cong \mIO$. 
\end{itemize}
If in addition the strict initial object $\mIO$ is \AP\emph{\intro{$\cM$-initial}}, i.e., if for all objects $X\in \obj{\bfC}$ the unique morphism $i_X:\mIO\rightarrow X$ is in $\cM$, then $\bfC$ has \AP\emph{\intro{disjoint coproducts}}, i.e., for all $X,Y\in \obj{\bfC}$, the pushout of the $\cM$-span $X\leftarrowtail\mIO\rightarrowtail Y$ is $X\rightarrowtail X+Y\leftarrowtail Y$ (cf.\ \cite[Thm.~3.2]{MONRO1986141}, which also states that this condition is equivalent to requiring $\bfC$ to be a \AP\emph{\intro{solid quasi-topos}}), and the coproduct injections are $\cM$-morphisms as well. Finally, if pushouts along regular monos of $\bfC$ are van Kampen, $\bfC$ is a \AP\emph{\intro{rm-adhesive category}}~\cite[Def.~1.1]{garner2012axioms}.
\end{corollary}

\subsection{The categories of directed multi- and simple graphs}\label{sec:graphs}

Throughout this paper, we will illustrate our constructions with two prototypical examples of (quasi-)topoi, namely categories of two types of directed graphs.

\begin{definition}\label{def:Graph}
The \AP\intro{category $\mathbf{Graph}$} of \emph{directed multigraphs} is defined as the presheaf category $\mathbf{Graph}:=(\bG^{op}\rightarrow \mathbf{Set})$, where $\bG:=(\cdot \rightrightarrows \star)$ is a category with two objects and two morphisms~\cite{ls2004adhesive}. Objects $G=(V_G,E_G,s_G,t_G)$ of $\mathbf{Graph}$ are given by a set of vertices $V_G$, a set of directed edges $E_G$ and the source and target functions $s_G,t_G:E_G\rightarrow V_G$. Morphisms of $\mathbf{Graph}$ between $G,H\in \obj{\mathbf{Graph}}$ are of the form $\varphi=(\varphi_V,\varphi_E)$, with $\varphi_V:V_G\rightarrow V_H$ and $\varphi_E:E_G\rightarrow E_H$ such that $\varphi_V\circ s_G=s_H\circ \varphi_E$ and $\varphi_V\circ t_G=t_H\circ \varphi_E$.
\end{definition}

\begin{definition}\label{def:SGraph}
The \AP\intro{category $\mathbf{SGraph}$} of \emph{directed simple graphs}\footnote{Some authors prefer to not consider directly the category $\mathbf{BRel}$, but rather define $\mathbf{SGraph}$ as some category equivalent to $\mathbf{BRel}$, where simple graphs are of the form $\langle V,E\rangle$ with $E\subseteq V\times V$. This is evidently equivalent to directly considering $\mathbf{BRel}$, whence we chose to not make this distinction in this paper.} is defined as the category of binary relations $\mathbf{BRel}\cong\mathbf{Set}\git\Delta$~\cite{quasi-topos-2007}. Here, $\Delta:\mathbf{Set}\rightarrow\mathbf{Set}$ is the pullback-preserving diagonal functor defined via $\Delta X:= X\times X$, and $\mathbf{Set}\git\Delta$ denotes the full subcategory of the slice category $\mathbf{Set}/\Delta$ defined via restriction to objects $m:X\rightarrow \Delta X$ that are monomorphisms. More explicitly, an object of $\mathbf{Set}\git\Delta$ is given by $S=(V,E, \iota)$, where $V$ is a set of vertices, $E$ is a set of directed edges, and where $\iota:E\rightarrow V\times V$ is an injective function. A morphism $f=(f_V,f_E)$ between objects $S$ and $S'$ is a pair of functions $f_V:V\rightarrow V'$ and $f_E:E\rightarrow E'$ such that $\iota'\circ f_E=(f_V\times f_V)\circ\iota$ (see~\eqref{eq:epiRMfactorizationSGraph}).
\end{definition}

These two categories satisfy the following well-known properties:
\begin{theorem}
The category $\mathbf{Graph}$ is an \emph{adhesive category} and (by definition) a \emph{presheaf topos}~\cite{ls2004adhesive} (and thus in particular a \kl{quasi-topos}), with strict-initial object $\mIO=(\emptyset, \emptyset,\emptyset\to\emptyset,\emptyset\to\emptyset)$ the empty graph, and with the following additional properties:
\begin{itemize}
\item Morphisms are in the classes $\mono{\mathbf{Graph}}$/$\epi{\mathbf{Graph}}$/$\iso{\mathbf{Graph}}$ if they are component-wise injective/surjective/bijective functions, respectively. All monos in $\mathbf{Graph}$ are regular, and $\mathbf{Graph}$ therefore possesses an epi-mono-factorization.
\item For each $G\in \obj{\mathbf{Graph}}$~\cite[Sec.~2.1]{Corradini_2015}, $\eta_G:G\rightarrow T(G)$ is defined as the embedding of $G$ into $T(G)$, where $T(G)$ is defined as the graph with vertex set $V_G':=V_G\uplus \{\star\}$ and edge set $E_G\uplus E_G'$. Here, $E_G'$ contains one directed edge $e_{n,p}: v_n\rightarrow v_p$ for each pair of vertices $(v_n,v_p)\in V_G'\times V_G'$.
\end{itemize}
The category $\mathbf{SGraph}$ is \emph{not} adhesive, but it is a quasi-topos~\cite{quasi-topos-2007}, and with the following additional properties:
\begin{itemize}
\item In $\mathbf{SGraph}$~\cite{quasi-topos-2007} (compare~\cite[Prop.~9]{BRAATZ2011246}), morphisms $f=(f_V,f_E)$ are monic (epic) if $f_V$ is monic (epic), while isomorphisms satisfy that both $f_V$ and $f_E$ are bijective. \emph{Regular} monomorphisms in $\mathbf{SGraph}$ are those for which $(\iota, f_E)$ is a pullback of $(\Delta(f_V),\iota')$~\cite[Lem.~14(ii)]{quasi-topos-2007}, i.e., a monomorphism is regular iff it is \emph{edge-reflecting}. As is the case for any \kl{quasi-topos}, $\mathbf{SGraph}$ possesses an epi-regular mono-factorization.
\item The regular mono-partial map classifier $(T,\eta)$ of $\mathbf{SGraph}$ is defined as follows~\cite[Ex.~28.2(3)]{adamek2006}: for every object $S=(V,E,\iota)\in \obj{\mathbf{SGraph}}$, 
\begin{equation}
T(S):= (V_{\star}=V\uplus \{\star\}, E_{\star}=E\uplus (V\times\{\star\})\uplus (\{\star\}\times V)\uplus \{(\star,\star)\}, \iota_{\star})\,,
\end{equation}
where $\iota_{\star}$ is the evident inclusion map, and moreover $\eta_{S}:S\rightarrowtail T(S)$ is the (by definition edge-reflecting) inclusion of $S$ into $T(S)$.
\item $\mathbf{SGraph}$ possesses a regular mono-initial object $\mIO=(\emptyset, \emptyset,\emptyset\to\emptyset)$.
\end{itemize}
\end{theorem}
\begin{proof}
While most of these results are standard, we briefly demonstrate that the epi-regular mono-factorization of $\mathbf{SGraph}$~\cite{quasi-topos-2007} is  ``inherited'' from the epi-mono-factorization of the adhesive category $\mathbf{Set}$. To this end, given an arbitrary morphism $f=(f_V,f_E)$ in $\mathbf{SGraph}$ as on the left of~\eqref{eq:epiRMfactorizationSGraph}, the epi-mono-factorization $f_V=m_V\circ e_V$ lifts via application of the diagonal functor $\Delta$ to a decomposition of the morphism $f_V\times f_V$. Pulling back $(\Delta(m_v),\iota')$ results in a span $(\tilde{\iota},f_E'')$ and (by the \kl{universal property of pullbacks}) an induced morphism $f_E'$ that makes the diagram commute. By stability of monomorphisms under pullbacks, $\tilde{\iota}$ is a monomorphism, thus the square marked $(*)$ precisely constitutes the data of a regular monomorphism in $\mathbf{SGraph}$, while the square marked $(\dag)$ is an epimorphism in $\mathbf{SGraph}$ (since $e_V\in \epi{\mathbf{Set}}$).
\end{proof}

\begin{equation}\label{eq:epiRMfactorizationSGraph}
\ti{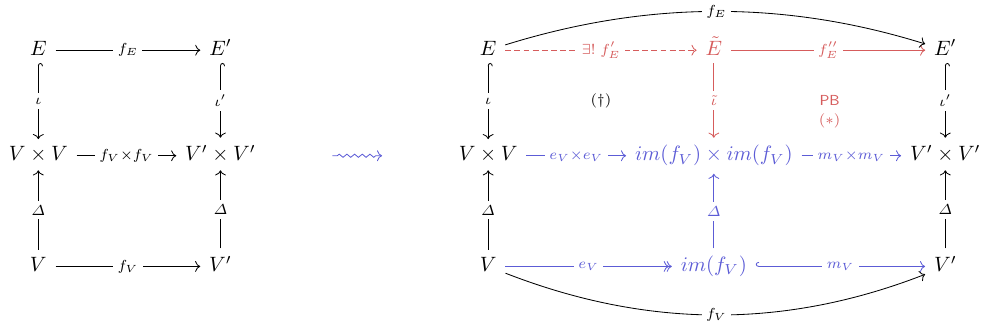}
\end{equation}

\subsection{FPCs, $\cM$-multi-POCs, $\cM$-multi-sums and FPAs}\label{sec:special}

Compared to compositional SqPO-type rewriting for $\cM$-linear rules~\cite{nbSqPO2019}, in the generic SqPO-type setting we require both a generalization of the concept of pushout complements that forgoes uniqueness, as well as a certain form of FPC-augmentation. To this end, it will prove useful to recall from~\cite{Corradini_2015} the following constructive result:

\begin{theorem}[\cite{Corradini_2015}, Thm.~1]\label{thm:FPC}
For a category $\bfC$ with \kl{$\cM$-partial map classifier} $(T,\eta)$, the \kl{final pullback complement (FPC)} of a composable sequence of arrows $A\xrightarrow{f}B$ and $B\xrightarrow{m}C$ with $m\in \cM$ is guaranteed to exist, and is constructed via the following algorithm:
\begin{enumerate}
\item Let $\bar{m}:=\varphi(m,id_B)$ (i.e., the morphism that exists by the universal property of $(T,\eta)$, cf.\ square $(1)$ below).
\item Construct $T(A)\xleftarrow{\bar{n}}F\xrightarrow{g}C$ as the pullback of $T(A)\xrightarrow{T(f)}T(B)\xleftarrow{\bar{m}}C$ (cf.\ square $(2)$ below); by the \kl{universal property of pullbacks}, this in addition entails the existence of a morphism $A\xrightarrow{n}F$.
\end{enumerate}
Then $(n,g)$ is the FPC of $(f,m)$, and $n$ is in $\cM$.
\end{theorem}

\begin{equation}\label{eq:M-FPC}
\ti{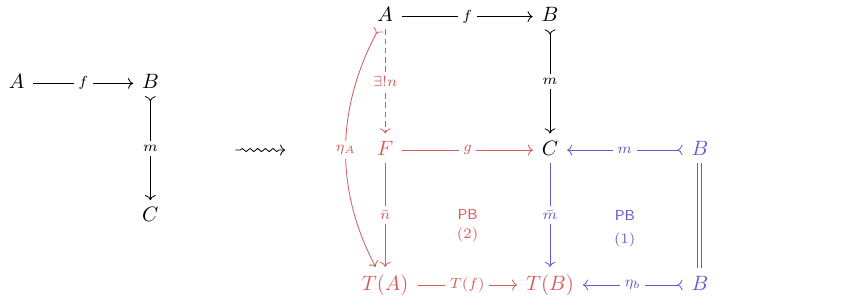}
\end{equation}

This guarantee for the existence of FPCs will prove quintessential for constructing $\cM$-multi-pushout complements, which are defined as follows:
\begin{definition}
For a category $\bfC$ with an \kl{$\cM$-partial map classifier}, the \emph{\intro{$\cM$-multi-pushout complement}} ($\cM$-multi-POC) $\mPOC{f}{b}$ of a composable sequence of morphisms $A\xrightarrow{f}B$ and $B\xrightarrow{b}D$ with $b\in \cM$ is defined as
\begin{equation}
\mPOC{f}{b}:= \{(A\xrightarrow{a}P, P\xrightarrow{d}D)\in \mor{\bfC}^{2}\mid a\in \cM \land (d,b) =\pO{a,f}\}\,.
\end{equation}
\end{definition}

\begin{proposition}
In a quasi-topos $\bfC$ and for $\cM=\regmono{\bfC}$ the class of regular monomorphisms, let $\mPOC{f}{b}$ be an \kl{$\cM$-multi-POC}.
\begin{itemize}
\item \AP\textbf{\intro{Universal property of $\mPOC{f}{b}$}}: for every diagram such as in~\eqref{eq:M-FPC-PO-augmentation-POCs}(i) where $(1)+(2)$ is a pushout along an $\cM$-morphism $n$, and where $m=m'\circ b$ for some $m',b\in \cM$, there exists an element $(a,d)$ of $\mPOC{f}{b}$ and an $\cM$-morphism $p\in \cM$ such that the diagram commutes and $(2)$ is a pushout. Moreover, for any $p'\in \cM$ and for any other element $(a',d')$of $\mPOC{f}{b}$ with the same property, there exists an isomorphism $\delta\in \iso{\bfC}$ such that $\delta\circ a=a'$ and $d'\circ \delta =d$.
\item \AP\textbf{\intro{Algorithm to compute $\mPOC{f}{b}$}}:
\begin{enumerate}
\item Construct $(n,g)$ in diagram~\eqref{eq:M-FPC-PO-augmentation-POCs}(ii) by taking the FPC of $(f,b)$. 
\item For every pair of morphisms $(a,p)$ such that $a\in \cM$ and $a\circ p=n$, take the pushout $(1)$, which by \kl{universal property of pushouts} induces an arrow $D\xrightarrow{e}C$; if $e\in \iso{\bfC}$, $(a,d)$ is a contribution to the \kl{$\cM$-multi-POC} of $(f,b)$.
\end{enumerate}
\end{itemize}
\begin{equation}\label{eq:M-FPC-PO-augmentation-POCs}
\ti{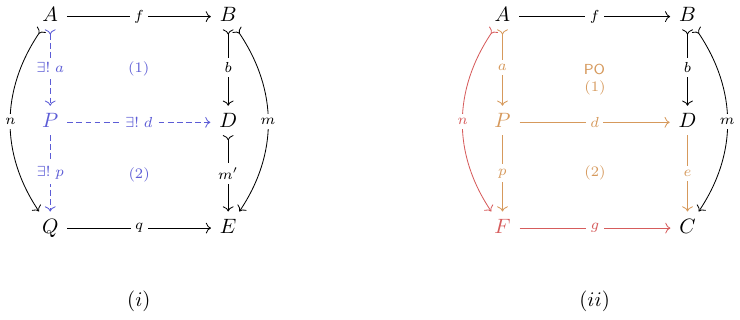}
\end{equation}
\end{proposition}
\begin{proof}
The universal property of $\mPOC{f}{b}$ follows from \kl{pushout-pullback decomposition}: \kl{pushouts along $\cM$-morphisms are pullbacks}, so $(1)+(2)$ is a pullback; taking the pullback $(p,d)$ of $(q,m')$ yields by the \kl{universal property of pullbacks} a morphism $a$ (which is unique up to isomorphism), and thus by \kl{pullback-pullback decomposition} that $(1)$ and $(2)$ are pullbacks. By \kl{stability of $\cM$-morphisms under pullbacks}, both $a$ and $p$ are in $\cM$, and finally by \kl{pushout-pullback decomposition}, both $(1)$ and $(2)$ are pushouts. This proves that $(a,d)$ is in $\mPOC{f}{b}$.

To prove that the algorithm provided indeed computes $\mPOC{f}{b}$, note first that by the \kl{universal property of FPCs}, whenever in a diagram as in~\eqref{eq:M-FPC-PO-augmentation-POCs}(ii) we have that $D\cong C$ and $b\in \cM$, since \kl{pushouts along $\cM$-morphisms are pullbacks}, square $(1)$ is a pullback, which entails by the \kl{universal property of FPCs} that there exists a morphism $p$ such that $p\circ a=n$. By \kl{stability of $\cM$-morphisms under pullbacks}, we find that $a$ must be in $\cM$, so indeed every possible contribution to $\mPOC{f}{b}$ must give rise to a diagram as in~\eqref{eq:M-FPC-PO-augmentation-POCs}(ii), which proves the claim. 
\end{proof}

An example of an $\cM$-multi-POC construction both in $\mathbf{SGraph}$ and in $\mathbf{Graph}$ is given in the diagram below. Note that in $\mathbf{Graph}$, the $\cM$-multi-POC does not contain the FPC contribution (since in $\mathbf{Graph}$ the pushout of the relevant span would yield to a graph with a multi-edge).
\begin{equation}
\vcenter{\hbox{\includegraphics[scale=0.3]{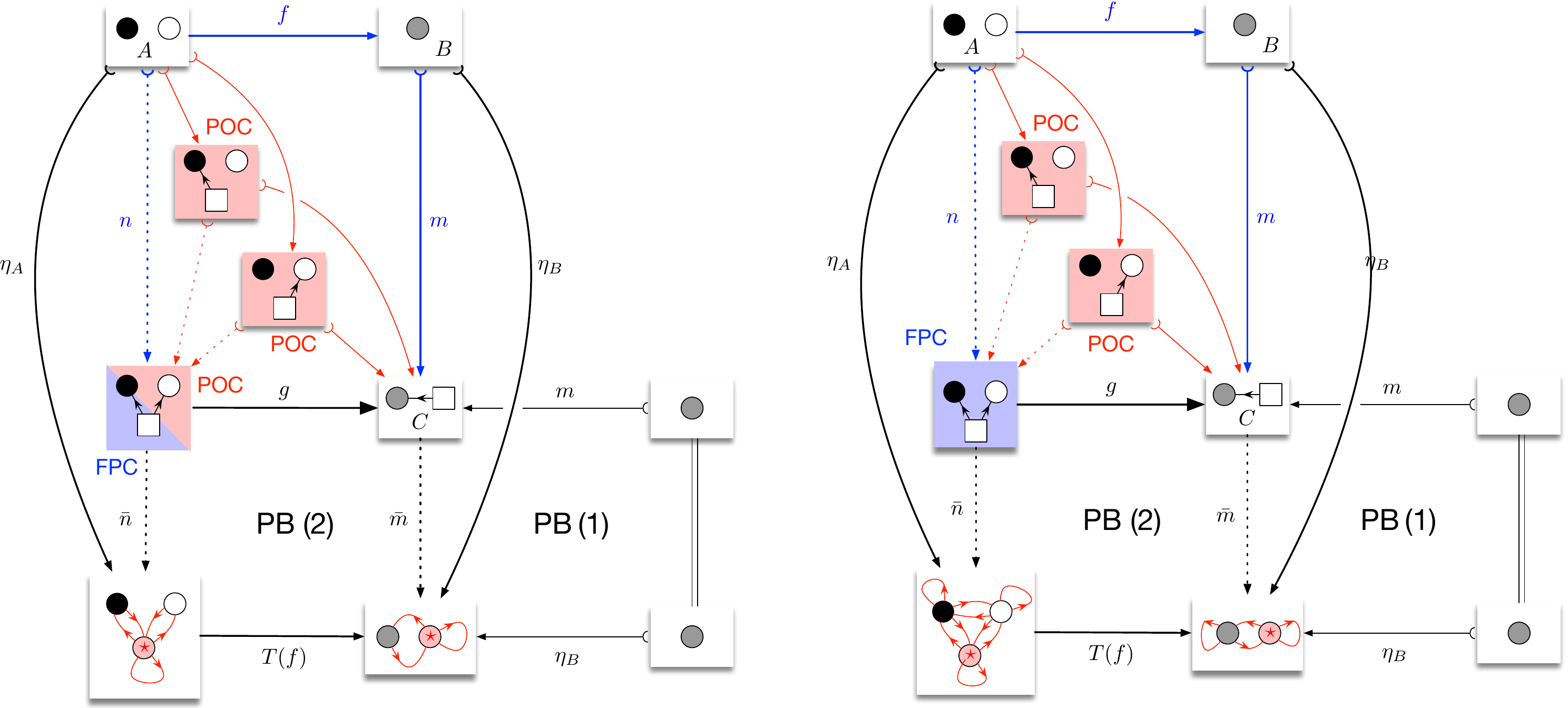}}}
\end{equation}

\begin{definition}[$\cM$-FPC-augmentations]
In a quasi-topos\footnote{As demonstrated in~\cite[Fact~3.4]{GABRIEL_2014}, every finitary $\cM$-adhesive category $\bfC$ possesses an (extremal $\cE$, $\cM$)-factorization, so if $\bfC$ is known to possess FPCs as required by the construction, this might allow to generalize the $\cM$-FPC-PO-augmentation construction to this setting.} $\bfC$ with $\cM=\regmono{\bfC}$, consider a pushout square along an $\cM$-morphism such as square $(1)$ in the diagram below (where $\alpha,\bar{\alpha}\in \cM$):
\begin{equation}\label{eq:defFPA}
\ti{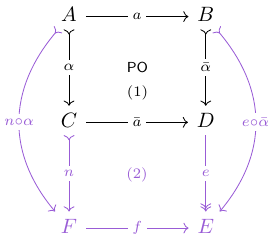}
\end{equation}
We define an \AP \intro{$\cM$-FPC augmentation (FPA)} of the pushout square $(1)$ as a diagram formed from an epimorphism $e\in \epi{\bfC}$ and that satisfies the following properties:
\begin{itemize}
\item The morphism $e\circ \bar{\alpha}$ is an $\cM$-morphism.
\item $(\bar{\alpha}, id_B)$ is a pullback of $(e,e\circ \bar{\alpha})$.
\item Square $(1)+(2)$ is an FPC, and the induced morphism $n$ that exists\footnote{Note that square $(1)$ pasted with the pullback square formed by the morphisms $\overline{\alpha}, id_B,e,e\circ \bar{\alpha}$ yields a pullback square that is indeed of the right form to warrant the existence of a morphism $n$ into the FPC square $(1)+(2)$.} by the \kl{universal property of FPCs}, here w.r.t.\ the FPC $(n\circ \alpha,f)$ of $(a,e\circ \bar{\alpha})$, is an $\cM$-morphism.
\end{itemize}
For a pushout as in $(1)$, we denote by $\FPA{\alpha}{a}$ its \AP\emph{\intro{class of FPAs}}:
\begin{equation}
\FPA{\alpha}{a}:=\{ (n,f,e)\mid  e\in \epi{\bfC} \land e\circ \bar{\alpha},n\in \cM\land (f,n\circ \alpha)=FPC(a,e\circ\bar{\alpha})\}
\end{equation}
As induced by the properties of pushouts and of FPCs, FPAs are defined up to universal isomorphisms (in $D$, $E$ and $F$), and for a given pushout square there will in general exist multiple non-isomorphic such augmentations.
\end{definition}

The final technical ingredient for our rewriting theoretic constructions is a notion of \emph{multi-sum} adapted to the setting of quasi-topoi,  a variation on the general theory of multi-(co-)limits due to Diers~\cite{diers1978familles}.

\begin{definition}\label{def:multisum}
In a quasi-topos $\bfC$, the \AP\emph{\intro{multi-sum}} $\Msum{A}{B}$ of two objects $A,B\in \obj{\bfC}$ is defined as a family of cospans of regular monomorphisms $A\xrightarrow{f}Y\xleftarrow{g}B$ with the following \AP\emph{\intro[multi-sum]{universal property}}: for every cospan $A\xrightarrow{a}Z\xleftarrow{b}B$ with $a,b\in \regmono{\bfC}$, there exists an element $A\xrightarrow{f}Y\xleftarrow{g}B$ in $\Msum{A}{B}$ and a regular monomorphism $Y\xrightarrow{y} Z$ such that $a=y\circ f$ and $b=y\circ g$, and moreover $(f,g)$ as well as $y$ are unique up to universal isomorphisms.
\end{definition}

\begin{equation}\label{eq:multiSumDef}
\ti{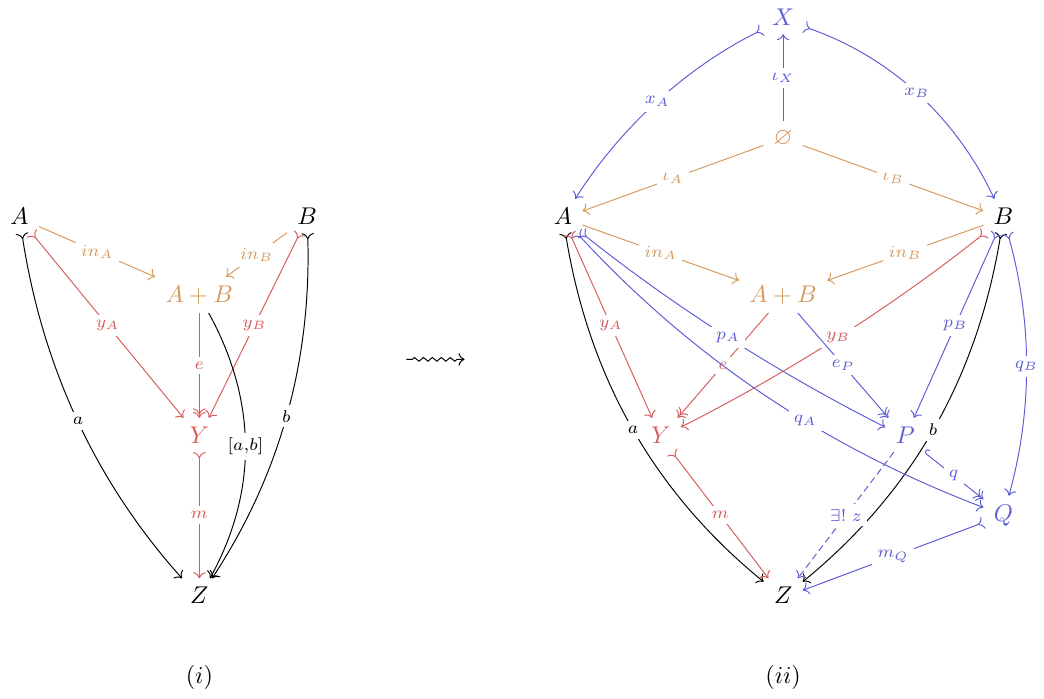}
\end{equation}

\begin{lemma}\label{lem:mSum}
If $\bfC$ is a \kl{quasi-topos}, the multi-sum $\Msum{A}{B}$ arises from the \kl{epi-$\cM$-factorization} of $\bfC$ (for $\cM=\regmono{\bfC}$; compare~\cite{harmer2020knowledge}). 
\begin{itemize}
\item \textbf{Existence:} Let $A\xrightarrow{in_A}A+B\xleftarrow{in_B}B$ be the disjoint union of $A$ and $B$. Then for any cospan $A\xrightarrow{a}Z\xleftarrow{b}B$ with $a,b\in \cM$, the epi-$\cM$-factorization of the induced arrow $A+B\xrightarrow{[a,b]}Z$ into an epimorphism $A+B\xrightarrow{e}Y$ and an $\cM$-morphism $Y\xrightarrow{m}Z$ yields a cospan $(y_A=e\circ in_A,y_B=e\circ in_B)$, which by the \kl{decomposition property of $\cM$-morphisms} is a cospan of $\cM$-morphisms (cf.~\eqref{eq:multiSumDef}(i)). 
\item \textbf{Construction:} For objects $A,B\in \obj{\bfC}$, every element $A\xrightarrow{q_A}Q\xleftarrow{q_B}B$ in $\Msum{A}{B}$ is obtained from a pushout of some span $A\xleftarrow{x_A}X\xrightarrow{x_B}B$ with $x_A,x_B\in \cM$ and a morphism $P\xrightarrow{q}Q$ in $\mono{\bfC}\cap\epi{\bfC}$ (cf. \eqref{eq:multiSumDef}(ii)).
\end{itemize}
\end{lemma}
\begin{proof}
See Appendix~\ref{app:multiSumLemmaProof}.
\end{proof}

\begin{wrapfigure}[11]{R}{0.23\textwidth}
\vspace{-2em}
    \includegraphics[scale=0.4]{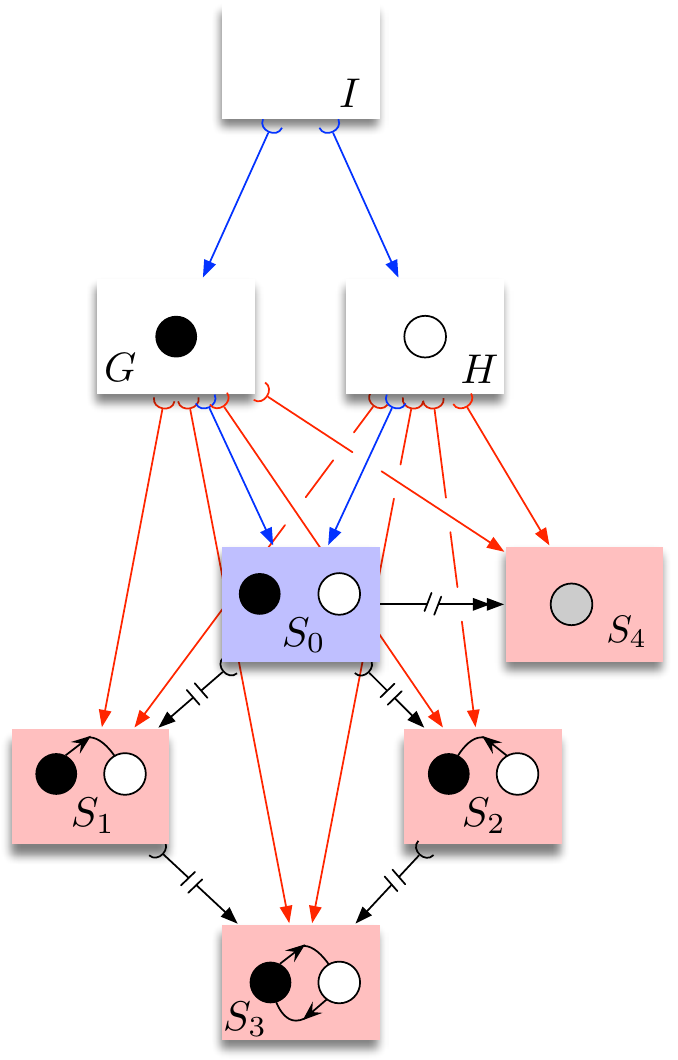}
\end{wrapfigure}
Since in an adhesive category all monos are regular~\cite{ls2004adhesive}, in this case the multi-sum construction simplifies to the statement that every monic cospan can be uniquely factorized as a cospan obtained as the pushout of a monic span composed with a monomorphism. It is however worthwhile emphasizing that for generic quasi-topoi $\bfC$ one may have $\cM\neq \mono{\bfC}$, as is the case in particular for the quasi-topos $\mathbf{SGraph}$ of simple graphs. We illustrate this phenomenon in the diagram on the right via presenting the multi-sum construction for $A=B=\bullet$. Note in particular the monic-epis that extend the two-vertex graph $S_0$ into the graphs $S_1$, $S_2$ and $S_3$, all of which have the same vertices as $S_0$ (recalling that a morphism in $\mathbf{SGraph}$ is monic/epic if it is so on vertices), yet additional edges, so that in particular none of the morphisms $S_0\rightarrow S_j$ for $j=1,2,3$ is edge-reflecting.

\section{Non-linear sesqui- and double-pushout rewriting}\label{sec:nlSqPO}

In much of the traditional work on graph- and categorical rewriting theories~\cite{ehrig:2006fund}, while it was appreciated early in its development that in particular SqPO-rewriting permits the cloning of subgraphs~\cite{Corradini_2006}, and that both SqPO- and DPO-semantics permit the fusion of subgraphs (i.e. via input-linear, but output-non-linear rules), the non-uniqueness of pushout complements along non-monic morphisms for the DPO- and the lack of a concurrency theorem in the SqPO-case in general has prohibited a detailed development of non-linear rewriting theories to date. Interestingly, the SqPO-type concurrency theorem for linear rules as developed in~\cite{nbSqPO2019} exhibits the same obstacle for the generalization to non-linear rewriting as the DPO-type concurrency theorem, i.e., the non-uniqueness of certain pushout complements. %
Our proof for non-linear rules identifies in addition a new and highly non-trivial  ``backpropagation effect'', which will be highlighted in Section~\ref{sec:cct}. It may be worthwhile emphasizing that there exists previous work that aimed at circumventing some of the technical obstacles of non-linear rewriting either via specializing the semantics e.g.\ from double pushout to a version based upon so-called \emph{minimal} pushout complements~\cite{BRAATZ2011246}, or from sesqui-pushout to \emph{reversible} SqPO-semantics~\cite{reversibleSqPO,harmer2020reversibility} or other variants such as AGREE-rewriting~\cite{Corradini_2015}. In contrast, we will in the following introduce the ``true'' extensions of both SqPO- and DPO-rewriting to the non-linear setting, with our constructions based upon multi-sums, multi-POCs and FPAs.

\begin{definition}\label{def:NLRSsqpo}
\AP\emph{\intro{General SqPO-rewriting semantics}} over a \kl{quasi-topos} $\bfC$:
\begin{itemize}
\item The \AP\emph{\intro{set of SqPO-admissible matches}} of a rule \emph{rule} $r=(O\leftarrow K\rightarrow I)\in span(\bfC)$ into an object $X\in \obj{\bfC}$ is defined as
\begin{equation}
\MatchGT{SqPO}{r}{X}:= \{ I\xrightarrow{m}X\mid m\in \regmono{\bfC}\}\,.
\end{equation}
A \AP\emph{\intro{SqPO-type direct derivation}}\footnote{Note that this part of the definition of general SqPO-semantics coincides precisely with the original definition of~\cite{Corradini_2006}.} of $X\in \obj{\bfC}$ with rule $r$ along $m\in \MatchGT{SqPO}{r}{X}$ is defined as a diagram in~\eqref{eq:defNLRSdd}, where $(1)$ is formed as an FPC, while $(2)$ is formed as a pushout.
\begin{equation}\label{eq:defNLRSdd}
\ti{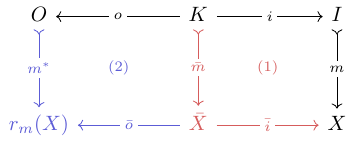}
\end{equation}
\item The \AP\emph{\intro{set of SqPO-type admissible matches of rules}}  $r_2,r_1\in span(\bfC)$ (also referred to in the literature as \emph{dependency relations}) is defined as
\begin{equation}
\begin{aligned}
\RMatchGT{SqPO}{r_2}{r_1}&:= \{
(j_2,j_1,\bar{j_1}, \bar{o}_1,\bar{\bar{j}}_1,\bar{\bar{i}}_1, \iota_{21})\mid \\
&\quad\quad (j_2,j_1)\in \Msum{I_2}{O_1} \land (\bar{j}_1, \bar{o}_1)\in \mPOC{o_1}{j_1}\\
&\quad\quad \land (\bar{\bar{j}}_1,\bar{\bar{i}}_1, \iota_{21})\in \FPA{\bar{j}_1}{i_1}
\}\diagup_{\sim}\,,
\end{aligned}
\end{equation}
where equivalence is defined up to the compatible universal isomorphisms of multi-sums, multi-POCs and FPAs (see below).
\item An \AP\emph{\intro{SqPO-type rule composition}} of two general rules $r_1,r_2\in span(\bfC)$ along an admissible match $\mu\in \RMatchGT{SqPO}{r_2}{r_1}$ is defined via a diagram as in~\eqref{eq:defNLRSrc} below, where (going column-wise from the left) squares $(2_2)$, $(6)$, and $(4)$ are pushouts, $(1_1)$ is the multi-POC element specified as part of the data of the match, $(2_1)$ and $(3)$ form an FPA-diagram as per the data of the match, and finally $(1_2)$ and $(5)$ are FPCs: 
\begin{equation}\label{eq:defNLRSrc}
\ti{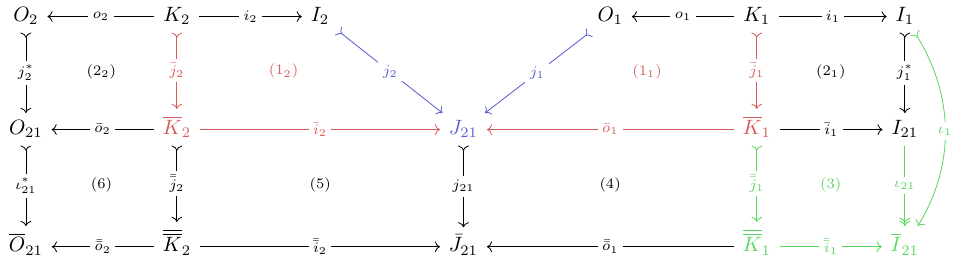}
\end{equation}
We then define the composite rule via span composition:
\begin{equation}
\sqComp{r_2}{\mu}{r_1}:= (\overline{O}_{21}\leftarrow\overline{\overline{K}}_{2} \rightarrow \overline{J}_{21})\circ (\overline{J}_{21}\leftarrow \overline{\overline{K}}_1\rightarrow \overline{I}_{21})
\end{equation}
\end{itemize}
\end{definition}

\begin{definition}
\AP\emph{\intro{General DPO-rewriting semantics}} over an \kl{rm-adhesive category} $\bfC$:
\begin{itemize}
\item The \AP\emph{\intro{set of DPO-admissible matches}} of a rule \emph{rule} $r=(O\leftarrow K\rightarrow I)\in span(\bfC)$ into an object $X\in \obj{\bfC}$ is defined as
\begin{equation}
\MatchGT{DPO}{r}{X}:= \{ (m,\bar{m},\bar{i})\mid m\in \regmono{\bfC} \land (\bar{m},\bar{i})\in \mPOC{i}{m}\}\,.
\end{equation}
A \AP\emph{\intro{DPO-type direct derivation}} of $X\in \obj{\bfC}$ with rule $r$ along $m\in \MatchGT{DPO}{r}{X}$ is defined as a diagram in~\eqref{eq:defNLRSdd}, where $(1)$ is the multi-POC element chosen as part of the data of the match, while $(2)$ is formed as a pushout.
\item The \AP\emph{\intro{set of DPO-type admissible matches of rules}}  $r_2,r_1\in span(\bfC)$ (also referred to as \emph{dependency relations}) is defined as
\begin{equation}
\begin{aligned}
\RMatchGT{DPO}{r_2}{r_1}&:= \{
(j_2,j_1,\bar{j}_2, \bar{i}_2,\bar{j}_1,\bar{o}_1)\mid \\
&\quad\quad (j_2,j_1)\in \Msum{I_2}{O_1}\\
&\quad\quad  \land (\bar{j}_2, \bar{i}_2)\in \mPOC{i_2}{j_2}\land (\bar{j}_1, \bar{o}_1)\in \mPOC{o_1}{j_1}
\}\diagup_{\sim}\,,
\end{aligned}
\end{equation}
where equivalence is defined up to the compatible universal isomorphisms of multi-sums and multi-POCs (see below).
\item A \AP\emph{\intro{DPO-type rule composition}} of two general rules $r_1,r_2\in span(\bfC)$ along an admissible match $\mu\in\RMatchGT{DPO}{r_2}{r_1}$  is defined via a diagram as in~\eqref{eq:defNLRDPOrc} below, where $(1_2)$ and $(1_1)$ are the multi-POC elements chosen as part of the data of the match, while $(2_2)$ and $(2_1)$ are pushouts: 
\begin{equation}\label{eq:defNLRDPOrc}
\ti{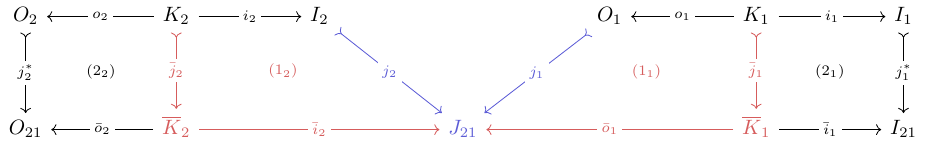}
\end{equation}
We then define the composite rule via span composition:
\begin{equation}
\comp{r_2}{\mu}{r_1}:= (O_{21}\leftarrow\overline{K}_{2} \rightarrow J_{21})\circ (J_{21}\leftarrow \overline{K}_1\rightarrow I_{21})
\end{equation}
\end{itemize}
\end{definition}
The precise reasons for the definitions of SqPO- and DPO-semantics for generic rules and regular monos as matches will only become evident via the concurrency theorems that will be developed in the following sections.

Let us illustrate the notion of SqPO-type rule composition, as given in Definition~\ref{def:NLRSsqpo}, with the following example in the setting of directed multi-graphs. Note that, since this is an adhesive category, all monos are automatically regular and we therefore have no need to restrict matches to being edge-reflecting monomorphisms.
\begin{equation}\label{eq:SqPOexample}
\vcenter{\hbox{\includegraphics[width=0.9\textwidth]{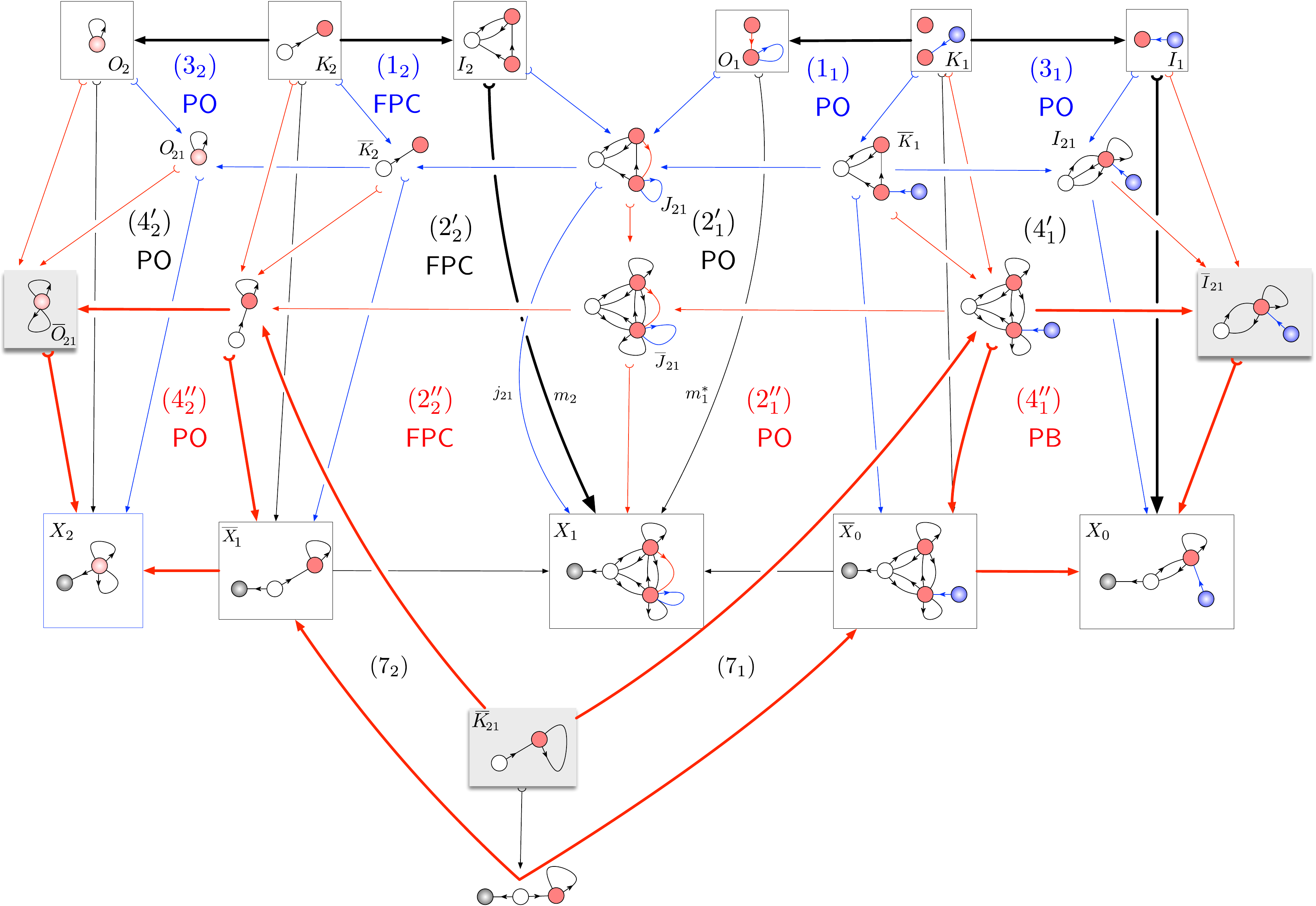}}}
\end{equation}

In this example, we have two rules. The first clones a node\footnote{Note that we have drawn the rule from right to left so that the \emph{input}, sometimes called the left-hand side, of the rule is the topmost rightmost graph. Note also that the structure of the homomorphisms may be inferred from the node positions, with the exception of the vertex clonings that are explicitly mentioned in the text.}, but not its incident edge, then adds a new edge between the original node and its clone and merges the blue node with the original node. The second rule deletes a node and then merges two nodes. The given applications to the graphs $X_0$ and $X_1$ illustrate some of the idiosyncrasies of SqPO-rewriting: 
\begin{itemize}
\item Since the node of $X_0$ that is being cloned possesses a self-loop, the result of cloning is two nodes, each with a self-loop, with one edge going each way between them. 
\item In the application of the second rule to $X_1$, we see the side-effect whereby all edges incident to the deleted node are themselves deleted (as also occurs in SPO-, but not in DPO-rewriting). 
\end{itemize}
The overall effect of the two rewrites can be seen in $X_2$; as usual, this depends on the overlap between the images of $O_1$ and $I_2$ in $X_1$. This overlap is precisely the multi-sum element $J_{21}$. Since our example is set in an adhesive category, this can be most easily computed by taking the PB of $m_1^*$ and $m_2$ and then the PO of the resulting span. The PO that defines the rewrite from $\overline{X}_0$ to $X_1$ can now be factorized by computing the PB of $j_{21}$ and the arrow from $\overline{X}_0$ to $X_1$; this determines $\overline{K}_1$ and its universal arrow from $K_1$ with consequence that $(1_1)$ and $(2_1)$ are both POs. Let us note that $\overline{K}_1$ is the appropriate member of the multi-POC, as determined by the particular structure of $\overline{X}_0$.

The PO $(3_1)$ induces a universal arrow from $I_{21}$ to $X_0$; but an immediate inspection reveals that this homomorphism is not a mono (nor an epi in this case). As such, we cannot hope to use $I_{21}$ as the input/left hand side of the composite rule. Furthermore, we find that the square $(4_1)$ is neither a PB nor a PO. However, the FPA $\overline{I}_{21}$ resolves these problems by enabling a factorization of this square, giving rise to a monomorphism $m_{21}$ into $X_0$, where $(4''_1)$ and $(3_1)+(4'_1)$ are PBs and indeed FPCs. This factorization, as determined by $e_{21}$, can now be \emph{back-propagated} to factorize $(2_1)$ into POs $(2'_1)$ and $(2''_1)$ which gives rise to an augmented version  $\overline{J}_{21}$ of the multi-sum object in the middle. Note moreover that the effect of back-propagation concerns also the contribution of the second rule in the composition: the final output motif contains an extra self-loop (compared to the motif $O_{21}$ defined by the PO $(3_2)$), which is induced by the extra self-loop of $\overline{J}_{21}$ that appears due to back-propagation.

We may then compute the composite rule via taking a pullback to obtain $\overline{K}_{21}$, yielding in summary the rule $\overline{O}_{21}\leftarrow\overline{K}_{21}\rightarrow\overline{I}_{21}$. Performing the remaining steps of the ``synthesis'' construction of the concurrency theorem (compare Appendix~\ref{sec:proofSynthesisSqPO}) then amounts to constructing the commutative cube in the middle of the diagram, yielding the FPC $(7_1)$ and the PO $(7_2)$, and thus finally the one-step SqPO-type direct derivation from $X_0$ to $X_2$ along the composite rule $\overline{O}_{21}\leftarrow\overline{K}_{21}\rightarrow\overline{I}_{21}$.\\

Let us finally note, as a general remark, that if the first rule in an SqPO-type rule composition is output- (or right-) linear then the POC is uniquely determined; and if it is input- (or left-) linear then the PO $(3_1)$ is also an FPC and $(4_1)$ is a PB, by Lemma 2(h) of~\cite{nbSqPO2019}. In this case, the FPA is trivial, and consequently so is the back-propagation process. Our rule composition can thus be seen as a conservative extension of that defined for linear rules in~\cite{nbSqPO2019}. 

\section{Concurrency theorem for non-linear SqPO rewriting}\label{sec:cct}

Part of the reason that a concurrency theorem for generic SqPO-rewriting had remained elusive in previous work concerns the intricate nature of the interplay between multi-sums, multi-POCs and FPAs as seen from the definition of rule compositions according to Definition~\ref{def:NLRSsqpo}, which is justified via the following theorem, constituting the first main result of the present paper:

\begin{theorem}
Let $\bfC$ be a \kl{quasi-topos}, let $X_0\in \obj{\bfC}$ be an object, and let $r_2,r_1\in span(\bfC)$ be two (generic) rewriting rules.
\begin{enumerate}
\item \textbf{Synthesis:} For every pair of admissible matches $m_1\in \MatchGT{SqPO}{r_1}{X_0}$ and $m_2\in \MatchGT{SqPO}{r_2}{r_{1_{m_1}}(X_0)}$, there exist an admissible match $\mu\in \RMatchGT{SqPO}{r_2}{r_1}$ and an admissible match $m_{21}\in \MatchGT{SqPO}{r_{21}}{X_0}$ (for $r_{21}$ the composite of $r_2$ with $r_1$ along $\mu$) such that $r_{21_{m_{21}}}(X_0)\cong r_{2_{m_2}}(r_{1_{m_1}}(X_0))$.
\item \textbf{Analysis:} For every pair of admissible matches $\mu\in \RMatchGT{SqPO}{r_2}{r_1}$ and $m_{21}\in \MatchGT{SqPO}{r_{21}}{X_0}$  (for $r_{21}$ the composite of $r_2$ with $r_1$ along $\mu$), there exists a pair of admissible matches $m_1\in \MatchGT{SqPO}{r_1}{X_0}$ and $m_2\in \MatchGT{SqPO}{r_2}{r_{1_{m_1}}(X_0)}$ such that $r_{2_{m_2}}(r_{1_{m_1}}(X_0))\cong r_{21_{m_{21}}}(X_0)$.
\item \textbf{Compatibility:} If in addition $\bfC$ is \AP\intro{finitary}~\cite[Def.~2.8]{GABRIEL_2014}, i.e., if for every object of $\bfC$ there exist only finitely many regular subobjects up to isomorphisms, the sets of pairs of matches $(m_1,m_2)$ and $(\mu,m_{21})$ are isomorphic if they are suitably quotiented by universal isomorphisms, i.e., by universal isomorphisms of $X_1=r_{1_{m_1}}(X_0)$ and $X_2=r_{2_{m_2}}(X_1)$ for the set of pairs $(m_1,m_2)$, and by the universal isomorphisms of multi-sums, multi-POCs and FPAs for the set of pairs $(\mu,m_{21})$, respectively.
\end{enumerate}
\end{theorem}
\begin{proof}
See Appendix~\ref{app:SqPOcctProof}
\end{proof}

\section{Concurrency theorem for non-linear DPO-rewriting}

The well-known and by now traditional results on concurrency in DPO-type semantics by Ehrig et al.\ were formulated for $\cM$-linear rules in $\cM$-adhesive categories (albeit possibly for non-monic matches; cf.\ \cite[Sec.~5]{ehrig:2006fund} for the precise details), and notably the non-uniqueness of pushout complements along non-linear morphisms posed the main obstacle for extending this line of results to non-linear DPO rewriting. As we will demonstrate in this section, taking advantage of multi-sums and multi-POCs, and if the underlying category $\bfC$ is an \kl{rm-adhesive category}~\cite[Def.~1.1]{garner2012axioms}, one may lift this restriction and obtain a fully well-posed semantics for DPO-rewriting along generic rules, and for regular monic matches:

\begin{theorem}\label{thm:CCTdpo}
Let $\bfC$ be an \kl{rm-adhesive category}, let $X_0\in \obj{\bfC}$ be an object, and let $r_2,r_1\in span(\bfC)$ be (generic) spans in $\bfC$.  
\begin{itemize}
\item \textbf{Synthesis:} For every pair of admissible matches $m_1\in \MatchGT{DPO}{r_1}{X_0}$ and $m_2\in \MatchGT{DPO}{r_2}{r_{1_{m_1}}(X_0)}$, there exist an admissible match $\mu\in \RMatchGT{DPO}{r_2}{r_1}$ and an admissible match $m_{21}\in \MatchGT{DPO}{r_{21}}{X_0}$ (for $r_{21}$ the composite of $r_2$ with $r_1$ along $\mu$) such that $r_{21_{m_{21}}}(X_0)\cong r_{2_{m_2}}(r_{1_{m_1}}(X_0))$.
\item \textbf{Analysis:} For every pair of admissible matches $\mu\in \RMatchGT{DPO}{r_2}{r_1}$ and $m_{21}\in \MatchGT{DPO}{r_{21}}{X_0}$  (for $r_{21}$ the composite of $r_2$ with $r_1$ along $\mu$), there exists a pair of admissible matches $m_1\in \MatchGT{DPO}{r_1}{X_0}$ and $m_2\in \MatchGT{SqPO}{r_2}{r_{1_{m_1}}(X_0)}$ such that $r_{2_{m_2}}(r_{1_{m_1}}(X_0))\cong r_{21_{m_{21}}}(X_0)$.
\item \textbf{Compatibility:} If in addition $\bfC$ is \kl{finitary}, the sets of pairs of matches $(m_1,m_2)$ and $(\mu,m_{21})$ are isomorphic if they are suitably quotiented by universal isomorphisms, i.e., by universal isomorphisms of $X_1=r_{1_{m_1}}(X_0)$ and $X_2=r_{2_{m_2}}(X_1)$ for the set of pairs of matches $(m_1,m_2)$, and by the universal isomorphisms of multi-sums and multi-POCs for the set of pairs of matches $(\mu,m_{21})$, respectively.
\end{itemize}
\end{theorem}
\begin{proof}
See Appendix~\ref{app:DPOcctProof}.
\end{proof}

It is worthwhile noting that for an adhesive category $\bfC$ (in which every monomorphism is regular) and if we consider \emph{linear} rules (i.e., spans of monomorphisms), the characterization of multi-sums according to Lemma~\ref{lem:mSum} permits to verify that DPO-type rule compositions as in Theorem~\ref{thm:CCTdpo} specialize in this setting precisely to the notion of DPO-type \emph{D-concurrent compositions }~\cite[Sec.~7.2]{lack2005adhesive}. This is because, in this case, each multi-sum element is precisely characterized as the pushout of a monic span (referred to as a D-dependency relation between rules in~\cite{lack2005adhesive}), so one finds indeed that Theorem~\ref{thm:CCTdpo} conservatively generalizes the traditional DPO-type concurrency theorem to the non-linear setting. Unlike for the generic SqPO-type setting however, quasi-topoi are \emph{not} sufficient for generic DPO-rewriting, since in the ``analysis'' part of the proof of the DPO-type concurrency theorem the van Kampen property of pushouts along regular monomorphisms is explicitly required (cf.\ Appendix~\ref{app:DPOcctProof}).

\section{Conclusion and outlook}

We have defined an abstract setting for SqPO graph transformation in quasi-topoi that captures the important concrete cases of (directed) multi-graphs and simple graphs. In particular, we have established the existence of appropriate notions of $\mathcal{M}$-multi-sums, $\mathcal{M}$-multi-POCs and $\mathcal{M}$-FPC-PO-augmentations in this setting that permit a proof of the concurrency theorem for general non-linear rules. %

Our immediate next goal is to prove associativity of our notion of rule composition in order to enable the use of rule algebra constructions~\cite{bdg2016,bp2019-ext,BK2020} and tracelets~\cite{behr2019tracelets} for static analysis~\cite{bdg2019,Behr2021} of systems generated by non-linear SqPO or DPO transformations. Intuitively, associativity is necessary in order to guarantee that one may consistently analyze and classify derivation traces based upon nested applications of the concurrency theorem, in the sense that recursive rule composition operations should yield a ``catalogue'' of all possible ways in which rules can interact in derivation sequences. The latter is formalized as the so-called tracelet characterization theorem in~\cite{behr2019tracelets}, whereby any derivation trace is characterized as an underlying tracelet and a match of the tracelet into the initial state of the trace. As illustrated in the worked example presented in~\eqref{eq:SqPOexample}, which highlighted the intriguing effect that comparatively complicated intermediate state in derivation traces involving cloning and fusing of graph structures are consistently abstracted away via performing rule compositions, one might hope that this type of effect persists also in $n$-step derivation traces for arbitrary $n$, for which however associativity is a prerequisite. Concretely, without the associativity property, the tentative ``summaries'' of the overall effects of derivation traces via their underlying tracelets would not be mathematically consistent, as they would only encode the causality of the nesting order in which they were calculated via pairwise rule composition operations. Preliminary results indicate however that indeed our generalized SqPO- and DPO-type semantics both satisfy the requisite associativity property, which will be presented in future work.

Beyond known applications to rule-based descriptions of complex systems, such as in Kappa~\cite{Boutillier:2018aa} and related formalisms, we hope to exploit this framework in graph combinatorics and structural graph theory~\cite{MBM2011}---which frequently employ operations such as edge contraction, which requires input-linear but output-non-linear rules, and node expansion, which further requires input-non-linear rules---to provide stronger tools for reasoning about graph reconfigurations as used, for example, in the study of coloring problems. We moreover expect this framework to be useful in strengthening existing approaches to graph-based knowledge representation~\cite{harmer2019bio}, particularly for the extraction and manipulation of audit trails~\cite{harmer2020reversibility} that provide a semantic notion of version control in these settings.

\clearpage
\appendix

\section{Collection of definitions and auxiliary properties}\label{sec:aux}

\subsection{Universal properties}\label{app:UP}

\begin{lemma}\label{lem:UP}
Let $\bfC$ be a category.
\begin{equation*}
\vcenter{\hbox{\ti{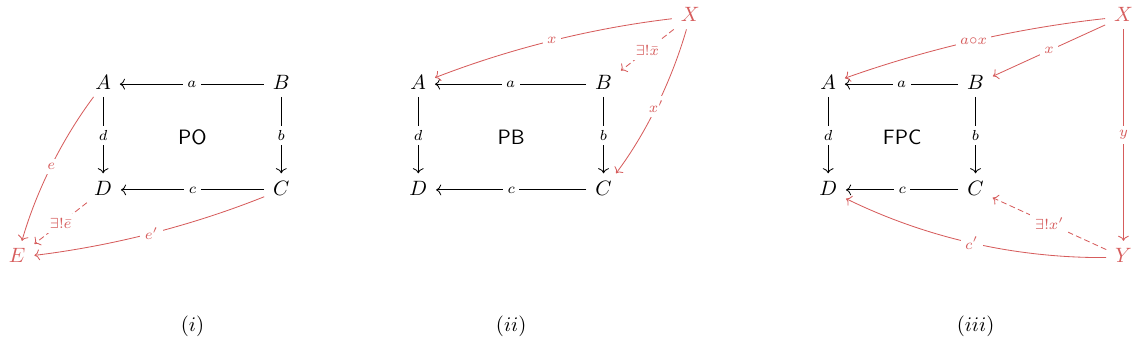}}}
\end{equation*}
Then the following properties hold:
\begin{enumerate}
\item \emph{\intro{Universal property of pushouts (POs)}}: Given a commutative diagram as in $(i)$, there exists a morphism $D-\bar{e}\to E$ that is unique up to isomorphisms.
\item \emph{\intro{Universal property of pullbacks (PBs)}}: Given a commutative diagram as in $(ii)$, there exists a morphism $X-\bar{x}\to A$ that is unique up to isomorphisms.
\item \emph{\intro{Universal property of final pullback complements (FPCs)}}: Given a commutative diagram as in $(iii)$ where $(a\circ x.y)$ is a PB of $(d,c')$, there exists a morphism $Y-\bar{x'}\to C$ that is unique up to isomorphisms, and which satisfies that $(x,y)$ is the PB of $(b,x')$. 
\end{enumerate}
\end{lemma}

\subsection{Stability properties}\label{sec:sp}

\begin{definition}\label{def:sPOFPC}
Let $\bfC$ be a category.
\begin{equation}\label{eq:stability}
\vcenter{\hbox{\ti{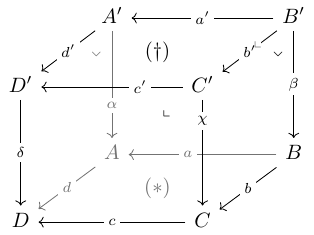}
}}
\end{equation}
\begin{itemize}
\item\AP A \intro{pushout $(*)$ in $\bfC$ is said to be \emph{stable under pullbacks}} iff for every commutative cube over the pushout $(*)$ such as in the diagram above where all vertical squares are pullbacks, the top square $(\dag)$ is a pushout.
\item\AP A \intro{final pullback complement (FPC) $(*)$ in $\bfC$ is said to be \emph{stable under pullbacks}} iff for every commutative cube over the FPC $(*)$ such as in the diagram above where all vertical squares are pullbacks, the top square $(\dag)$ is an FPC.
\end{itemize}
\end{definition}

\begin{lemma}\label{lem:poSP}
Two important examples of categories for which suitable stability properties for pushouts hold are given as follows:
\begin{enumerate}
\item In every \emph{adhesive category} $\bfC$, pushouts along monomorphisms (i.e., pushouts such as $(*)$ in~\eqref{eq:stability} with $a\in \mono{\bfC}$ or $b\in\mono{\bfC}$) are stable under pullback~\cite{ls2004adhesive}. This property is indeed the ``if'' direction of the so-called van Kampen property of adhesive categories~\cite{garner2012axioms}.
\item In a \AP\emph{\intro{regular mono (rm)-quasiadhesive category}}~\cite[Def.~1.1 and Cor.~4.7]{garner2012axioms}, all pushouts along regular monomorphisms exist, these pushouts are also pullbacks, and in particular \AP \intro{pushouts along regular monomorphisms are stable under pullbacks}. A useful characterization of rm-quasiadhesive categories is the following: a small category $\bfC$ with all pullbacks and with pushouts along regular monomorphisms is rm-quasiadhesive iff it has a full embedding into a quasi-topos (preserving the aforementioned two properties). 
\end{enumerate}
\end{lemma}

\begin{lemma}[\cite{reversibleSqPO}, Lem.~1]\label{lem:FPCs}
\AP Let $\bfC$ be a category with all pullbacks. Then \intro{FPCs are stable under pullbacks}.
\end{lemma}

\begin{proposition}\label{prop:qtt}
In a \kl{quasi-topos} $\bfC$,  \AP \intro{unions of regular subobjects are effective}~\cite[Prop.~10]{quasi-topos-2007}, i.e., the union of two subobjects is computed as the pushout of their intersection, and moreover the following property holds: in a commutative diagram such as below, where $(c,a)$ is the pullback of $(h,p)$, $(d,b)$ the pushout of $(c,a)$, where all morphisms (except possibly $x$) are monomorphisms, and where either $p\in \regmono{\bfC}$ or $h\in \regmono{\bfC}$, then the induced morphism $x:D\rightarrow E$ is a \emph{monomorphism}~\cite[Prop.~2.4]{garner2012axioms}:
\begin{equation}
\vcenter{\hbox{\ti{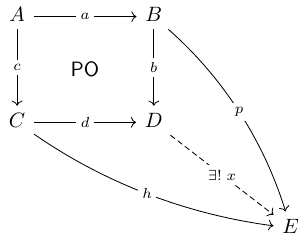}}}
\end{equation}
\end{proposition}

\subsection{Single-square lemmata specific to $\cM$-adhesive categories}\label{app:SSL}

\begin{lemma}\label{lem:dsl}
Let $\bfC$ be an $\cM$-adhesive category.
\begin{equation}
\vcenter{\hbox{\ti{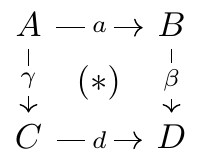}}}
\end{equation}
\begin{enumerate}
\item \AP \emph{\intro{Pushouts along $\cM$-morphisms are pullbacks}}: if $(*)$ is a pushout and $\beta\in \cM$, then $(*)$ is also a pullback.
\item \AP \emph{\intro{Stability of $\cM$-morphisms under pushouts}}: if $(*)$ is a pushout and $\beta\in \cM$, then $\gamma\in \cM$.
\item \AP \emph{\intro{Stability of $\cM$-morphisms under pullbacks}}: if $(*)$ is a pullback and $\gamma\in \cM$, then $\beta\in \cM$.
\item If $(*)$ is a pullback, $\gamma=id_A$ and $a,\beta\in \cM$, then $a\in \cM$.
\end{enumerate}
Since $(*)$ for $\gamma=id_A$ and $\beta\in \cM\subset \mono{\bfC}$ is always a pullback, 4. may be reformulated as follows:
\begin{enumerate}
\item[4.'] \AP \emph{\intro{Decomposition property of $\cM$-morphisms}}: if $g\circ f\in \cM$ and $g\in \cM$, then $f\in \cM$.
\end{enumerate} 
\end{lemma}

\subsection{Double-square lemmata}\label{app:DSL}

\begin{lemma}\label{lem:dsl}
Let $\bfC$ be a category.
\begin{equation}
\vcenter{\hbox{\ti{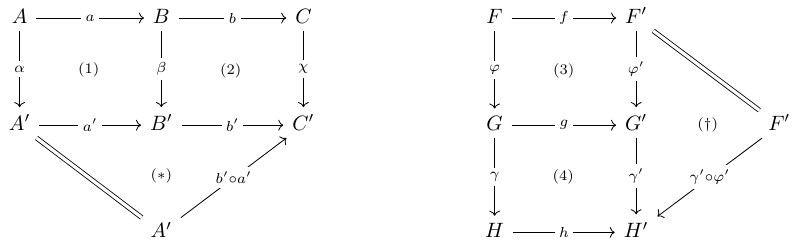}\gdef\mycdScale{1}}}
\end{equation}\label{eq:lemDSL}
Given commutative diagrams as above, the following statements hold:
\begin{enumerate}
\item \AP \emph{\intro{Pushout-pushout-(de-)composition}}: if $(1)$ is a pushout, $(1)+(2)$ is a pushout iff $(2)$ is a pushout.
\item \AP\emph{\intro{Pullback-pullback-(de-)composition}}: if $(2)$ is a pullback, $(1)+(2)$ is a pullback iff $(1)$ is a pullback.
\item \AP \emph{\intro{Pushout-pullback-decomposition}}~\cite[Lem.~4]{reversibleSqPO}: if $(1)+(2)$ is a \textbf{stable pushout}\footnote{Here, ``stable'' refers to stability under pullbacks.} and $(1)$, $(2)$, $(*)$ are pullbacks, then $(1)$ and $(2)$ are both pushouts. (Note: If $a'$ and $b'$ are monomorphisms, the condition on $(*)$ is always satisfied.)
\item \AP\emph{\intro{Pullback-pushout-decomposition}} (variant of~\cite[Lem.~B.2]{GOLAS2014}): if $\chi$ is in $\cM$, $(1)+(2)$ is a pullback and $(1)$ is \textbf{stable pushout}, then $(1)$ and $(2)$ are both pullbacks.
\item \AP \emph{\intro{Horizontal FPC-FPC-(de-)composition}}: if $(2)$ is an FPC (i.e., $(\beta,b')$ is an FPC of $(b,\chi)$), $(1)+(2)$ is an FPC iff $(1)$ is an FPC.
\item \AP \intro{Vertical FPC-FPC-(de-)composition}~\cite[Prop.~36]{Loewe_2015}: if $(3)$ is an FPC (i.e., $(\varphi,g)$ is an FPC of $(f,\varphi')$), 
\begin{enumerate}
\item if $(4)$ is an FPC (i.e., $(\gamma,h)$ is an FPC of $(g, \gamma')$), then $(3)+(4)$ is an FPC (i.e., $(\gamma\circ \varphi,h)$ is an FPC of $(f,\gamma'\circ\varphi')$)
\item if $(3)+(4)$ is an FPC (i.e., $(\gamma\circ \varphi,h)$ is an FPC of $(f,\gamma'\circ\varphi')$) and if $(4)$ is a pullback, then $(3)$  is an FPC (i.e., $(\varphi,g)$ is an FPC of $(f,\varphi')$).
\end{enumerate}
\item \AP \emph{\intro{Vertical FPC-pullback decomposition}}~\cite[Lem.~3]{reversibleSqPO}: if $(3)+(4)$ is an FPC (i.e., $(\gamma\circ \varphi,h)$ is an FPC of $(f,\gamma'\circ\varphi')$), both $(4)$ and $(\dag)$ are pullbacks, and if the diagram commutes, then $(3)$  is an FPC (i.e., $(\varphi,g)$ is an FPC of $(f,\varphi')$) and $(4)$ is an FPC (i.e., $(\gamma,h)$ is an FPC of $(g, \gamma')$). (Note: If $\gamma'$ and $\varphi'$ are monomorphisms, the condition on $(\dag)$ is always satisfied.)
\end{enumerate}
\end{lemma}

\begin{proof}
Most of the above results are cited from previous works (with references provided), yet the \kl{pullback-pushout decomposition} statement is a slight generalization of the variant provided in~\cite[Lem.~B.2]{GOLAS2014} and thus requires a proof. Construct the diagram below via taking a pullback:
\begin{equation}
\vcenter{\hbox{\ti{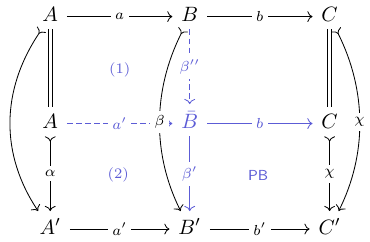}\gdef\mycdScale{1}}}
\end{equation}
\end{proof}
\begin{itemize}
\item The existence of the morphisms $A-a'\to \bar{B}$ and $B-\beta''\to \bar{B}$ follows from the \kl{universal property of pullbacks}.
\item By \kl{stability of $\cM$-morphisms under pullbacks}, $\beta'$ is in $\cM$, and thus by the \kl{decomposition property of $\cM $-morphisms}, $\beta''$ is also in $\cM$. The latter entails that $(1)$ is a pullback.
\item By \kl{pullback-pullback decomposition}, $(2)$ is a pullback.
\item Since $(1)$ and $(2)$ are pullbacks, $(1)+(2)$ a pushout along a $\cM$-morphism (i.e., a stable pushout), and since $\beta',\beta''\in \cM$, by \kl{pushout-pullback decomposition}, both $(1)$ and $(2)$ are pushouts. Therefore, $B\cong B'$, and the claim follows.
\end{itemize}

\section{Proof of Lemma~\ref{lem:mSum}}\label{app:multiSumLemmaProof}

The only non-trivial part about the existence statement concerns the fact that $y_A$ and $y_B$ are in $\cM$, which follows from the \kl{decomposition property of $\cM $-morphisms}. As for the construction of multi-sum elements, let us first prove that $A+B\xrightarrow{e_P}P$ is indeed an epimorphism. To this end, for a cospan of $\cM$-morphisms $A\xrightarrow{p_A}P\xleftarrow{p_B}B$ obtained via pushout of some span of $\cM$-morphisms $A\xleftarrow{x_A}X\xrightarrow{x_B}B$, let $A+B\xrightarrow{f}P$ for $f=[p_A,p_B]$ denote the induced morphism, and let $A+B\xrightarrow{e_f}P'\xrightarrow{m_F}P$ denote the epi-$\cM$-factorization of $f$. Taking pullbacks along  $m_f$ to obtain the squares marked  $(1_A)$ and $(1_B)$ in~\eqref{eq:MsumProof}(i), which by the \kl{universal property of pullbacks} entails the existence of the morphisms marked $a''$ and $b''$, by \kl{stability of $\cM$-morphisms under pullbacks}, $a'$ and $b'$ are in $\cM$. Thus $(a'',a',id_A,id_A)$ and $(b'',b',id_A,id_A)$ are pullbacks, whence by stability of isomorphisms under pullbacks, $a',a'',b',b''$ are isomorphisms. Form the square marked $(2)$ as a pullback, which 
 by the \kl{universal property of pullbacks} also yields morphisms $\mIO\xrightarrow{\iota_X'} X'$ and $X'\xrightarrow{x}X$. By \kl{pullback-pullback decomposition}, all squares of the bottom commutative cube are pullbacks, so that by stability of isomorphisms under pullbacks, $X'\cong X$. The bottom-most square is a pushout along $\cM$-morphisms and thus a stable pushout, whence by \kl{stability of stable pushouts under pullbacks}, the square marked $(2)$ is a pushout. Thus by the \kl{universal property of pushouts}, $P'\cong P$, and we have proved that $A+B\xrightarrow{f}P$ is indeed an epimorphism (henceforth referred to as $e_P$).
 
 To proceed, denote by $P\xrightarrow{e_Q}Q\xrightarrow{m_Q}$ (with $e_Q\in \epi{\bfC}$ and $m_Q\in \cM$) the epi-$\cM$-factorization of the morphism $P\rightarrow Z$ that exists by the \kl{universal property of pushouts}. By uniqueness of epi-$\cM$-factorizations up to isomorphisms~\cite[Prop.~14.4]{adamek2006}, since $A+B\xrightarrow{e_Q\circ e_P}P\xrightarrow{m_Q}Z$ and $A+B\xrightarrow{e}Y\xrightarrow{m}Z$ are epi-$\cM$-factorizations of $A+B\xrightarrow{[a,b]}Z$, we find that $P\cong Y$. Finally, since $m_Q\in \cM$, the squares marked $(3_A)$ and $(3_B)$ in~\eqref{eq:MsumProof}(ii) are pullbacks. Forming the square marked $(4)$ as a pullback, by \kl{pullback-pullback decomposition} also the two back vertical squares in~\eqref{eq:MsumProof}(ii) are pullbacks, whence by stability of isomorphisms under pullbacks, $X'\cong X$. Since $P$ is a pushout of $\cM$-morphisms and since $(4)$ is a pullback along $\cM$-morphisms, it follows from \kl{effectiveness of binary unions of regular subobjects} that $P\xrightarrow{q}Q$ is a monomorphism, which proves the claim that $e_Q\in \mono{\bfC}\cap \epi{\bfC}$.
\begin{equation}\label{eq:MsumProof}
\ti{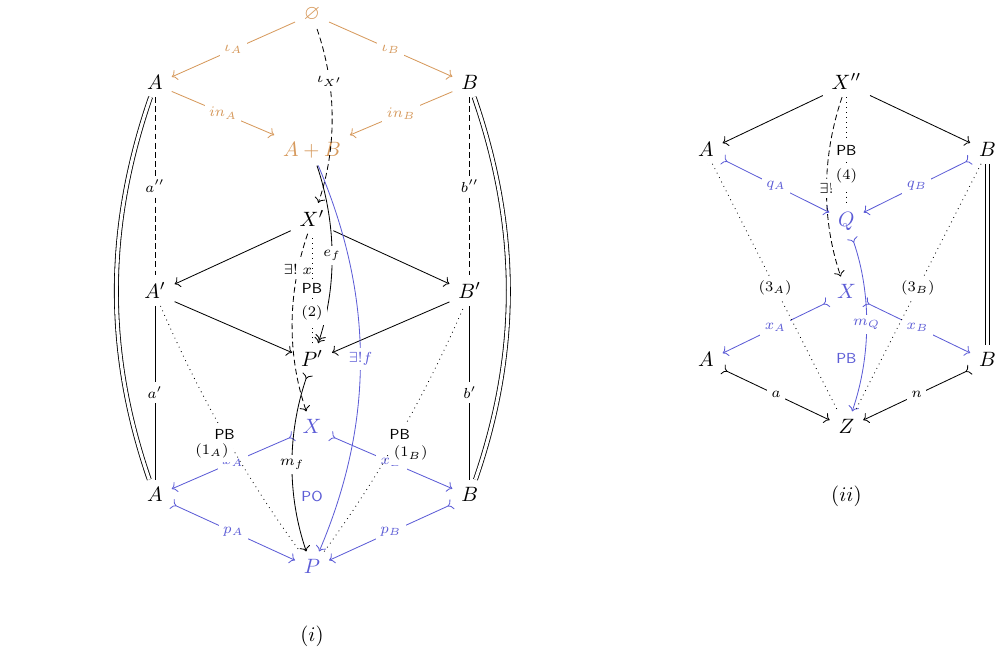}
\end{equation}

\section{Proof of the SqPO-type concurrency theorem}\label{app:SqPOcctProof}

Recall that we assume $\bfC$ is a \kl{quasi-topos}, and for the \emph{compatibility} part of the theorem in addition that $\bfC$ is \kl{finitary}.

\subsection{``Synthesis'' part}\label{sec:proofSynthesisSqPO}

Let $X_0\in \obj{\bfC}$ be an object, $r_j=(O_j\leftarrow K_j\rightarrow I_j)\in span(\bfC)$ ($j=1,2$) generic rules, and let $(m_1:I_1\rightarrowtail X_0)\in \RMatchGT{SqPO}{r_1}{X_0}$ and $(m_2:I_2\rightarrowtail X_1)\in \RMatchGT{SqPO}{r_2}{X_1}$ be $SqPO$-admissible matches, where $X_1:=r_{1_{m_1}}(X_0)$. %
Consider then a sequence of $SqPO$-type direct derivations, which yields a diagram as presented in~\eqref{eq:cctProofSynthesisStep1}, %
and  identify the \emph{multi-sum element} $(I_2\rightarrowtail J_{21}\leftarrowtail O_1)$, which is in particular a cospan of $\cM$-morphisms, and unique up to isomorphisms. By the \kl{universal property of multi-sums}, there exists an $\cM$-morphism $J_{21}\rightarrowtail X_1$:
\begin{equation}\label{eq:cctProofSynthesisStep1}
\ti{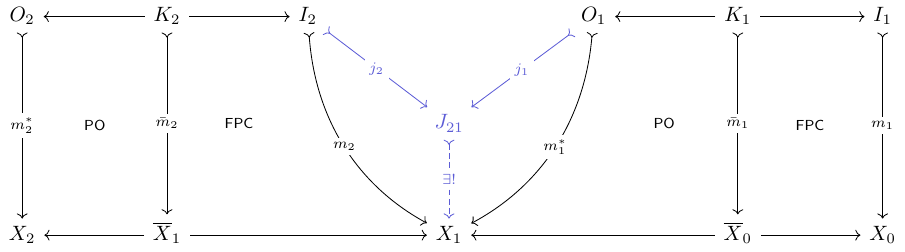}\gdef\mycdScale{1}
\end{equation}

Take the pullback $(J_{21}\leftarrow \bar{K}_1\rightarrow \bar{X}_0)$ of $(J_{21}\rightarrow X_1\leftarrow \bar{X}_0)$, and the pullback $(\bar{X}_1\leftarrow \bar{K}_2\rightarrow  J_{21})$ of $(\bar{X}_1\rightarrow X_1\leftarrow J_{21})$, resulting in the following diagram:
\begin{equation}
\ti{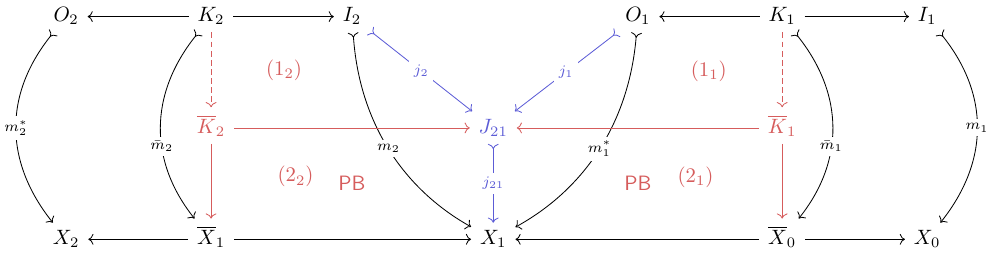}\gdef\mycdScale{1}
\end{equation}

\begin{itemize}
\item By \kl{stability of $\cM$-morphisms under pullbacks}, $(\bar{K}_1\rightarrow \bar{X}_0), (\bar{K}_2\rightarrow \bar{X}_0)\in\cM$.
\item By the \kl{universal property of pullbacks}, there exist the morphisms $K_1\rightarrow \bar{K}_1$ and $K_2\rightarrow \bar{K}_2$.
\item By the \kl{decomposition property of $\cM$-morphisms}, $(K_1\rightarrow \bar{K}_1),(K_2\rightarrow \bar{K}_2)\in \cM$.
\item Since by assumption $(O_1\rightarrowtail X_1\leftarrow \bar{X}_0)$ is the pushout of $(O_1\leftarrow K_1\rightarrowtail \bar{X}_0)$, and since \kl{pushouts along $\cM$-morphisms are pullbacks}, invoking \kl{pullback-pullback decomposition} yields that $(O_1\leftarrow K_1\rightarrowtail \bar{K}_1)$ is a pullback of $(O_1\rightarrowtail J_{21}\leftarrow \bar{K}_1)$. A completely analogous argument reveals that $(\bar{K}_2\leftarrowtail K_2\rightarrow I_2)$ is a pullback of $(\bar{K}_2\rightarrow J_{21}\leftarrowtail I_2)$.
\item Since moreover $(O_1\rightarrowtail J_{21})\in \cM$ and $(J_{21}\rightarrowtail X_1)\in \cM$, so that in particular the square $(1_1)+(2_1)$ is a pushout that is stable under pullbacks, by \kl{pushout-pullback decomposition} the squares $(1_1)$ and $(2_1)$ are also pushouts. 
\item Since $(I_2\rightarrowtail J_{21})\in \cM$ and $(J_{21}\rightarrowtail X_1)\in \cM$, since the square $(1_2)+(2_2)$ is an FPC, and since $(1_2)$ and $(2_2)$ are pullbacks, by \kl{vertical FPC-pullback decomposition} the squares $(1_2)$ and $(2_2)$ are FPCs. 
\end{itemize}

Next, form the squares marked $(3_1)$ and $(3_2)$ in the diagram below by taking pushouts:
\begin{equation}
\ti{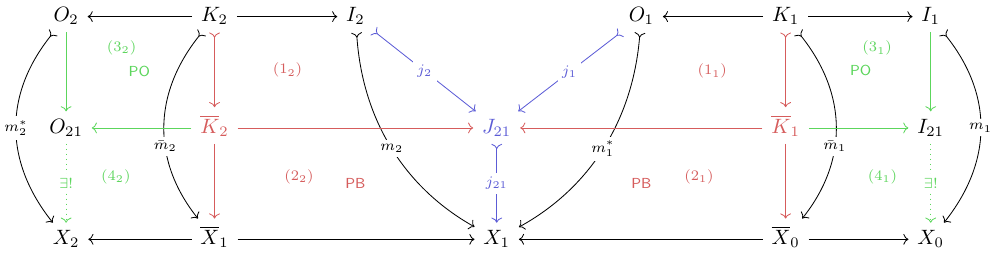}\gdef\mycdScale{1}
\end{equation}
\begin{itemize}
\item By \kl{stability of $\cM$-morphisms under pushouts}, $(I_1\rightarrow I_{21}), (O_2\rightarrow O_{21})\in\cM$.
\item Since $(3_2)+(4_2)$ and $(3_2)$ are pushouts, by \kl{pushout-pushout decomposition} $(4_2)$ is a pushout, and since moreover $\bar{K}_2\rightarrowtail \bar{X}_1$ is in $\cM$, by \kl{stability of $\cM$-morphisms under pushouts} we find that $(O_{21}\rightarrow X_2)\in \cM$.
\end{itemize}

In order to analyze the structure of the induced squares $(3_1)$ and $(4_1)$ in further detail, let us invoke the \kl{epi-$\cM$-factorization} of the morphism $I_{21}\rightarrow X_0$, and construct the following diagram:
\begin{equation}
\ti{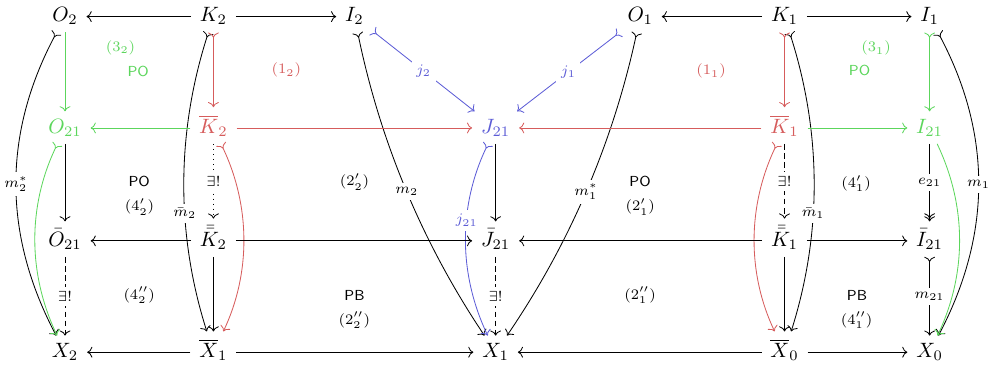}\gdef\mycdScale{1}
\end{equation}
\begin{itemize}
\item Since $(4_1'')$ is constructed as a pullback and $m_{21}\in \cM$, by \kl{stability of $\cM $-morphisms under pullbacks} $\bar{\bar{K}}_1\rightarrow \bar{X}_0$ is in $\cM$. 
\item By the \kl{decomposition property of $\cM $-morphisms}, the morphisms $I_1\rightarrow \bar{I}_{21}$, $\bar{K}_1\rightarrow \bar{\bar{K}}_1$ and $K_1\rightarrow \bar{\bar{K}}_1$ are all in $\cM$.
\item By \kl{vertical FPC-pullback decomposition}, both $(4_1'')$ and $(4_1)+(4_1')$ are FPCs.
\item By \kl{pushout-pushout decomposition}, since $(2_1')$ is constructed as a pushout and $(2_1')+(2_1'')$ is a pushout, $(2_1'')$ is also a pushout.
\item By \kl{stability of $\cM$-morphisms under pushouts}, the morphisms $J_{21}\rightarrow\bar{J}_{21}$ and $\bar{J}_{21}\rightarrow  X_1$ are in $\cM$. 
\item Since $(2_2'')$ is constructed as a pullback and since $(2_2')+(2_2'')$ is an FPC (and thus a pullback), by \kl{pullback-pullback decomposition} $(2_2')$ is a pullback.
\item By \kl{stability of $\cM$-morphisms under pullbacks}, the morphisms $\bar{K}_2\rightarrow \bar{\bar{K}}_2$ and $\bar{\bar{K}}_2\rightarrow \bar{X}_1$ are in $\cM$.
\item By \kl{vertical FPC-pullback decomposition}, $(2_2')$ and $(2_2'')$ are both FPCs.
\item Since $(4_2')$ is constructed as a pushout and since $(4_2')+(4_2'')$ is a pushout, by \kl{pushout-pushout decomposition} $(4_2'')$ is a pushout.
\item By \kl{stability of $\cM$-morphisms under pushouts}, the morphisms $O_{21}\rightarrow \bar{O}_{21}$ and $\bar{O}_{21}\rightarrow X_2$ are in $\cM$.
\end{itemize}

As an intermediate summary, we have thus derived the following information:
\begin{equation}
\ti{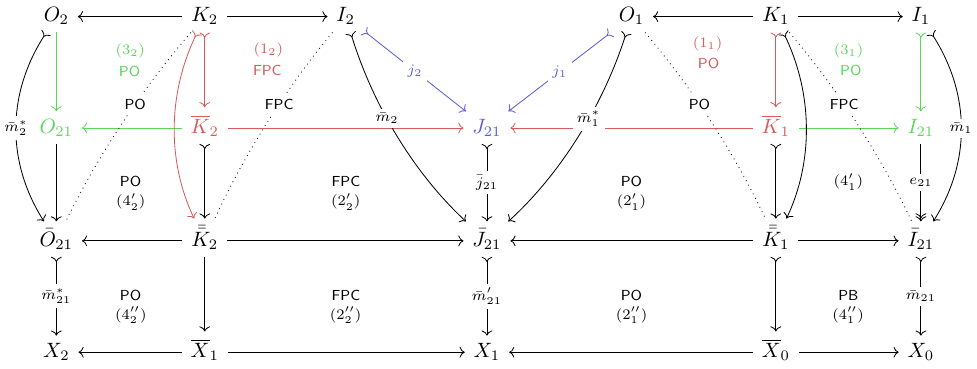}\gdef\mycdScale{1}
\end{equation}
\begin{itemize}
\item As indicated via the dotted lines, the vertical composition of the top two rows of the diagram yields a two-step sequence of SqPO-type direct derivations from $\bar{I}_{21}$ along rule $(O_1\leftarrow K_1\rightarrow I_1)$ with match $\bar{m}_1$ into $\bar{J}_{21}$, and then by rule $(O_2\leftarrow K_2\rightarrow I_2)$ with match $\bar{m}_1$ to $\bar{O}_{21}$.
\item The data of squares $(3_1)$ and $(4_1')$ furnishes an $\cM$-FPC augmentation (i.e., of the pushout square $(3_1')$ via morphisms $I_1\rightarrow \bar{I}_{21}$, $I_{21}\rightarrow\bar{I}_{21}$ and $\bar{K}_1\rightarrow \bar{\bar{K}}_1$).
\end{itemize}

To proceed, form the squares $(5)$ and $(6)$ in the diagram below by taking pullbacks:
\begin{equation}\label{eq:cctProofSynthesisStep5}
\ti{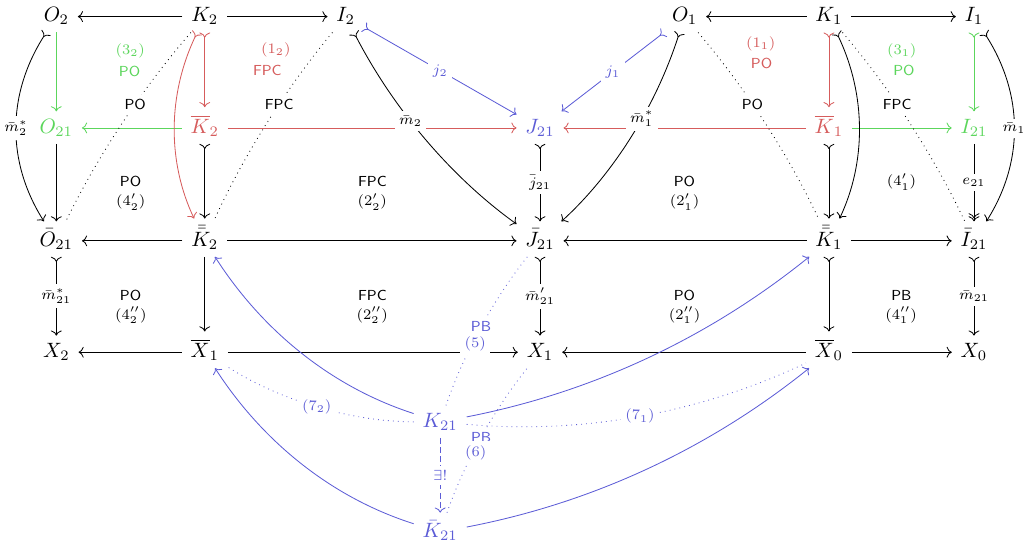}\gdef\mycdScale{1}
\end{equation}
\begin{itemize}
\item By the \kl{universal property of pullbacks}, there exists a morphism $K_{21}\rightarrow \bar{K}_{21}$.
\item By \kl{pullback-pullback decomposition}, since $(5)+(2_1'')$ and $(6)$ are pullbacks, $(7_2)$ is a pullback, and analogously since $(5)+(2_2'')$ is a pullback and $(6)$ is a pullback, $(7_1)$ is a pullback.
\item By \kl{stability of $\cM$-morphisms under pullbacks}, the morphism $K_{21}\rightarrow \bar{K}_{21}$ is in $\cM$.
\item Since $(2_1'')$ is a pushout along an $\cM$-morphism and thus stable under pullbacks, $(7_2)$ is a pushout. Thus by \kl{pushout composition}, $(7_2)+(4_2'')$ is a pushout.
\item Since $(2_2'')$ is an FPC and \kl{FPCs are stable under pullbacks}, $(7_1)$ is an FPC. Thus by \kl{horizontal FPC composition}, $(7_1)+(4_1'')$ is an FPC.
\end{itemize}
This concludes the proof of the ``synthesis'' part of the concurrency theorem, since the latter two points exhibit the data of a single-step SqPO-type direct derivation (of $X_0$ along $(\bar{O}_{21}\leftarrow K_{21}\rightarrow \bar{I}_{21})$ along match $\bar{m}_{21}$).

\subsection{``Analysis'' part}

Suppose we were given an SqPO-type composite rule as defined via the data in the diagram below:
\begin{equation}
\ti{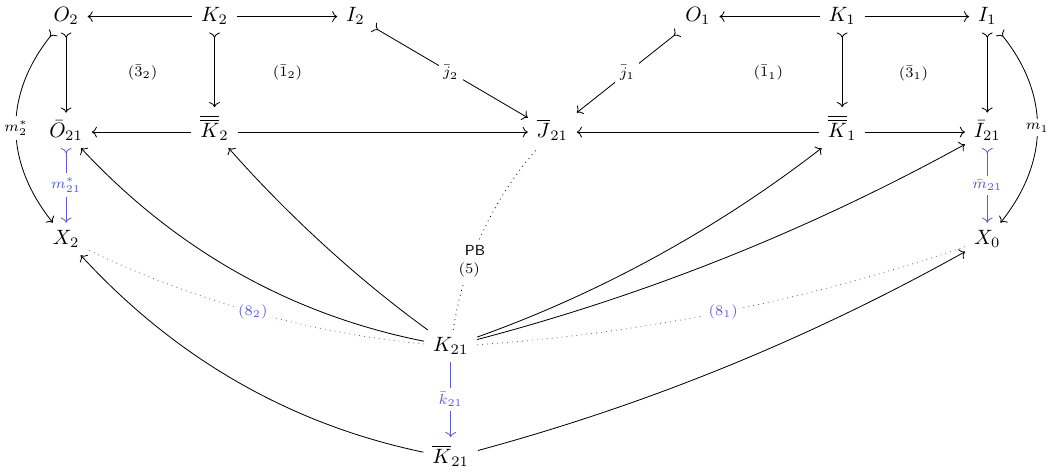}\gdef\mycdScale{1}
\end{equation}
Here, compared to the diagram in~\eqref{eq:cctProofSynthesisStep5}, we have for brevity only explicitly depicted the vertical compositions of the top two rows in~\eqref{eq:cctProofSynthesisStep5} (i.e., $(\bar{3}_2)$ and $(\bar{1}_1)$ are pushouts, while $(\bar{1}_2)$ and $(\bar{3}_1)$ are FPCs). According to the definition of SqPO-type direct derivations, we furthermore are given that $(8_1)$ is an FPC and $(8_2)$ a pushout.

Extend this diagram by forming FPC $(9)$, pushout $(10)$ and FPC $(11)$:
\begin{equation}\label{eq:cctAnalysisProofStep1}
\vcenter{\hbox{\ti{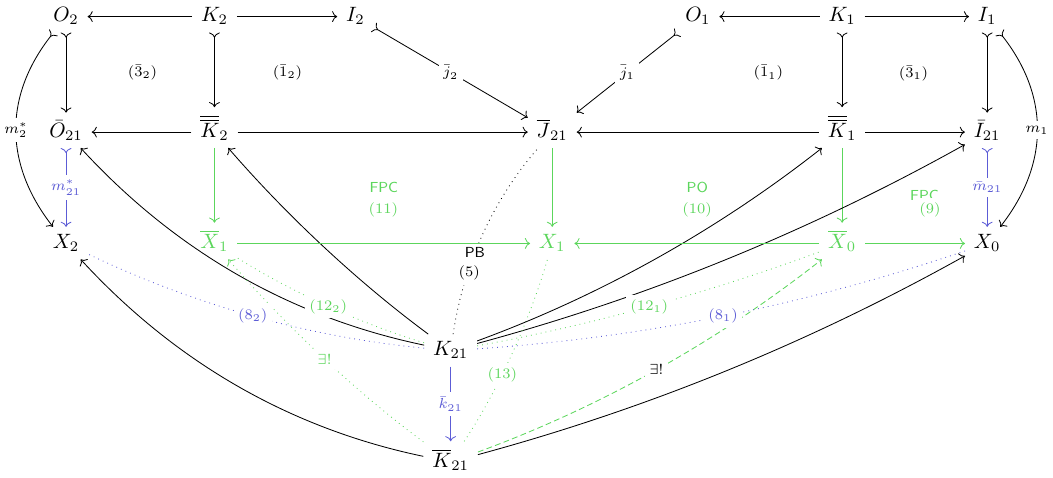}}}
\end{equation}
\begin{itemize}
\item By the \kl{universal property of FPCs}, there exist morphisms $\overline{K}_{21}\rightarrow \bar{X}_0$ and $\overline{K}_{21}\rightarrow \bar{X}_1$.
\item By \kl{stability of $\cM$-morphisms under pullbacks} and \kl{under pushouts}, respectively, the morphisms $\overline{\overline{K}}_1\rightarrow \bar{X}_0$,  $\overline{J}_{21}\rightarrow X_1$ and $\overline{\overline{K}}_2\rightarrow \bar{X}_1$ are in $\cM$.
\item Since $(8_1)=(12_1)+(9)$ and $(9)$ are FPCs, by \kl{horizontal FPC decomposition}, $(12_1)$ is an FPC.
\end{itemize}

To proceed, we have to construct the following auxiliary diagram, where the square $(14)$ on the bottom of the diagram is formed by taking the pullback $(\bar{X}_1\leftarrow Y\rightarrow \bar{X}_0)$ of the cospan $(\bar{X}_1\rightarrow X_1\leftarrow \bar{X}_0)$:
\begin{equation}\label{eq:cctAnalysisProofStep2}
\vcenter{\hbox{\ti{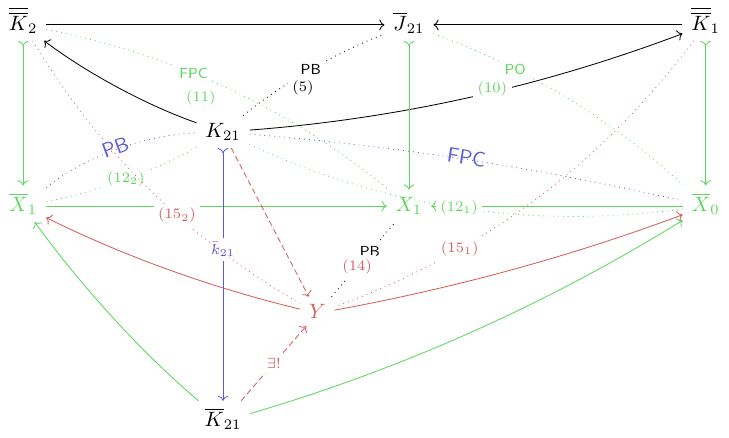}}}
\end{equation}
\begin{itemize}
\item By the \kl{universal property of pullbacks}, there exist morphisms $K_{21}\rightarrow Y$ and $\bar{K}_{21}\rightarrow Y$.
\item By \kl{pullback-pullback decomposition}, since $(5)+(11)$ and $(14)$ are pullbacks, $(15_1)$ is a pullback, and since moreover $(5)+(10)$ and $(14)$ are pullbacks, $(15_2)$ is a pullback.
\item Since $(10)$ is a pushout along an $\cM$-morphism, and since by assumption $\cM=\regmono{\bfC}$, by \kl{stability of pushouts along regular monomorphisms under pullbacks}, $(15_2)$ is a pushout.
\item By \kl{stability of FPCs under pullbacks}, $(15_1)$ is an FPC.
\item Since thus $(15_1)$ and $(12_1)$ are both FPCs of $K_{21}\rightarrow I_{21}\rightarrow \bar{X}_0$, by the \kl{universal property of FPCs} we find that $Y\cong \bar{K}_{21}$, and thus also that $(15_2)$ is a pushout, and the outer bottom square (which is the square labeled $(13)$ in~\eqref{eq:cctAnalysisProofStep1}) is a pullback.
\end{itemize}

Back to the diagram in~\eqref{eq:cctAnalysisProofStep1}, the previous auxiliary results entail in particular that since $(15_2)$ is a pushout and the diagram commutes, by the \kl{universal property of pushouts} there exists a morphism $\overline{X}_1\rightarrow X_2$ (yielding the square marked $(16)$ in the diagram below):
\begin{equation}\label{eq:cctAnalysisProofStep3}
\vcenter{\hbox{\ti{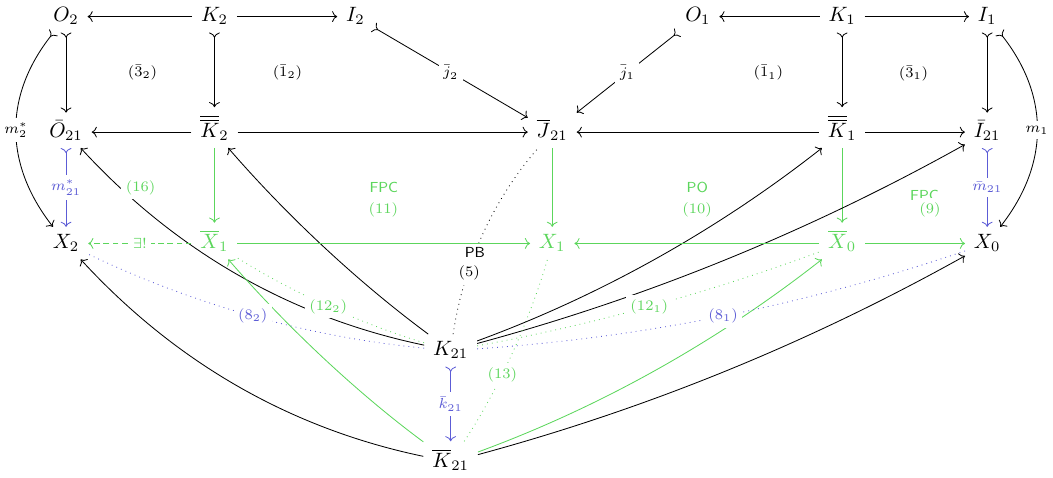}}}
\end{equation}
\begin{itemize}
\item Since $(15_2)$ and $(8_2)=(16)+(15_2)$ are pushouts, by \kl{pushout-pushout decomposition} the square $(16)$ is a pushout.
\item By \kl{vertical FPC composition} and by \kl{pushout composition}, respectively, $(\bar{3}_1)+(9)$ and $(\bar{1}_2)+(11)$ are FPCs, while $(\bar{1}_1)+(10)$ and $(\bar{3}_2)+(16)$ are pushouts.
\end{itemize}
This concludes the ``analysis'' part of the SqPO-type concurrency theorem, since the aforementioned statements amount to the structure of a sequence of SqPO-type direct derivations.

\subsection{``Compatibility'' part}

In order to prove that the set of pairs of SqPO-admissible matches of rules into objects on the one hand, and the set of pairs of SqPO-admissible rule compositions and SqPO-admissible matches of composite rules into objects on the other hand are isomorphic under a suitable notion of quotients, note first that the commutative diagrams constructed in the final steps of each of the ``analysis'' and of the ``synthesis'' parts of the proof are precisely equivalent. The data of a sequence of two SqPO-type direct derivations is uniquely defined up to universal isomorphisms originating from the \kl{universal property of FPCs} and the \kl{universal property of pushouts}. From each such two-step sequence, the ``synthesis'' construction permits to extract a SqPO-type direct derivation of a particular SqPO-type composite rule along an SqPO-admissible match, where this construction is again unique up to universal isomorphisms arising from the \kl{universal property of multi-sums}, the \kl{universal property of pullbacks} and the \kl{universal property of FPAs}. Quotienting the aforementioned sets of matches by these isomorphisms then permits to establish the postulated bijection.

\section{Proof of the DPO-type concurrency theorem}\label{app:DPOcctProof}

The proof of the theorem follows in essence the strategy for the proof of the traditional version for D-concurrent compositions as put forward in~\cite[Sec.~8.2]{lack2005adhesive} (compare~\cite{ehrig:2006fund}), adapted to the setting of non-linear rules via a few important modifications, namely via the multi-sum and multi-POC constructions. Moreover, since we will require the van Kampen property at a crucial point in the ``analysis'' part of the proof, the theorem statement includes the assumption that the underlying category $\bfC$ must be an \kl{rm-adhesive category} (rather than merely a \kl{quasi-topos} as in the non-linear SqPO setting).

\subsection{``Synthesis'' part}

Let $X_0\in \obj{\bfC}$ be an object, $r_j=(O_j\leftarrow K_j\rightarrow I_j)\in span(\bfC)$ ($j=1,2$) generic rules (i.e., spans of generic morphisms), and let $(m_1,x_0)\in \RMatchGT{DPO}{r_1}{X_0}$ and $(m_2, x_1)\in \RMatchGT{DPO}{r_2}{X_1}$ be $DPO$-admissible matches, where $X_1:=r_{1_{m_1}}(X_0)$. %
The first steps of this part of the proof are completely analogous to the corresponding proof in the SqPO setting (Section~\ref{sec:proofSynthesisSqPO}). %
Consider then a sequence of DPO-type direct derivations, which yields a diagram as presented in~\eqref{eq:cctProofSynthesisStep1}, and  identify (uniquely up to isomorphism) the element $(I_2\rightarrowtail J_{21}\leftarrowtail O_1)$ of the \kl{multi-sum} (which is in particular a cospan of $\cM$-morphisms), which by the \kl{universal property of multi-sums} entails the existence of an $\cM$-morphism $J_{21}\rightarrowtail X_1$:
\begin{equation}\label{eq:cctDPOproofSynthesisStep1}
\ti{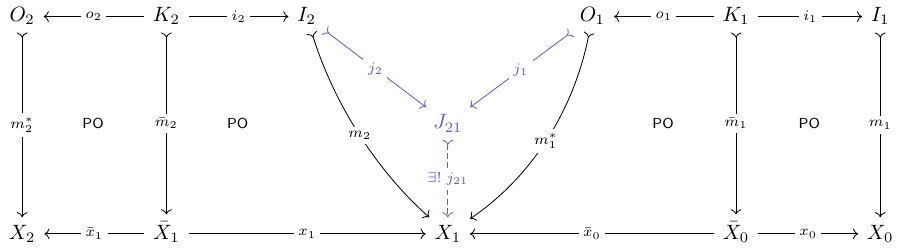}
\end{equation}

Take the pullback $(J_{21}\leftarrow \bar{K}_1\rightarrow \bar{X}_0)$ of $(J_{21}\rightarrow X_1\leftarrow \bar{X}_0)$, and the pullback $(\bar{X}_1\leftarrow \bar{K}_2\rightarrow  J_{21})$ of $(\bar{X}_1\rightarrow X_1\leftarrow J_{21})$, resulting in the following diagram:
\begin{equation}
\ti{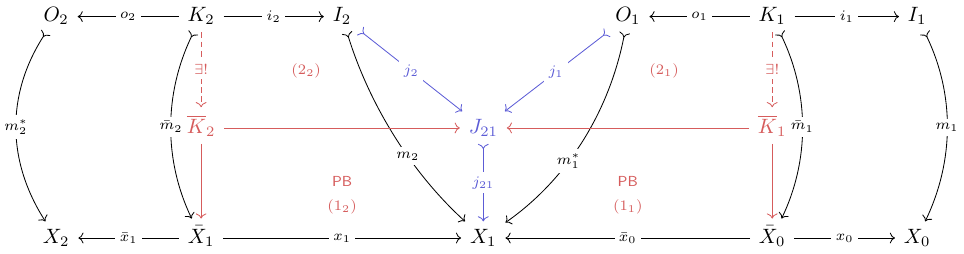}
\end{equation}

\begin{itemize}
\item By \kl{stability of $\cM$-morphisms under pullbacks}, $(\bar{K}_1\rightarrow \bar{X}_0), (\bar{K}_2\rightarrow \bar{X}_0)\in\cM$.
\item By the \kl{universal property of pullbacks}, there exist the morphisms $K_1\rightarrow \bar{K}_1$ and $K_2\rightarrow \bar{K}_2$.
\item By the \kl{decomposition property of $\cM$-morphisms}, $(K_1\rightarrow \bar{K}_1),(K_2\rightarrow \bar{K}_2)\in \cM$.
\item Since $(1_1)+(2_1)$ and $(1_2)+(2_2)$ are both pushouts along $\cM$-morphisms and thus also pullbacks (since \kl{pushouts along $\cM$-morphisms are pullbacks}), by \kl{pullback-pullback decomposition} both $(2_1)$ and $(2_2)$ are pullbacks. Since moreover \kl{pushouts along $\cM$-morphisms are stable under pullbacks}, by \kl{pushout-pullback decomposition} the squares $(1_1)$, $(1_2)$, $(2_1)$ and $(2_2)$ are all pushouts.
\end{itemize}

Next, form the squares marked $(3_1)$ and $(3_2)$ in the diagram below by taking pushouts:
\begin{equation}
\ti{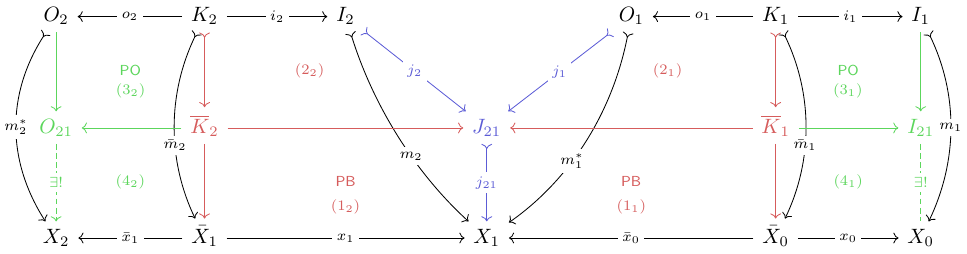}
\end{equation}
\begin{itemize}
\item By \kl{stability of $\cM$-morphisms under pushouts}, $(I_1\rightarrow I_{21}), (O_2\rightarrow O_{21})\in\cM$.
\item Since $(3_k)+(4_k)$ and $(3_k)$ are pushouts (for $k=1,2$), by \kl{pushout-pushout decomposition} $(4_1)$ and $(4_2)$ are pushouts, and since moreover $\bar{K}_1\rightarrowtail \bar{X}_0$ and $\bar{K}_2\rightarrowtail \bar{X}_1$ are in $\cM$, by \kl{stability of $\cM$-morphisms under pushouts} we find that $(I_{21}\rightarrow X_0)$ and $(O_{21}\rightarrow X_2)$ are also in $\cM$.
\end{itemize}

For the next step, let us first of all make explicit all information on the nature of all squares and morphisms constructed thus far, and let us construct in addition the squares marked $(5)$ and $(6)$ in the diagram below by taking pullbacks:
\begin{equation}
\ti{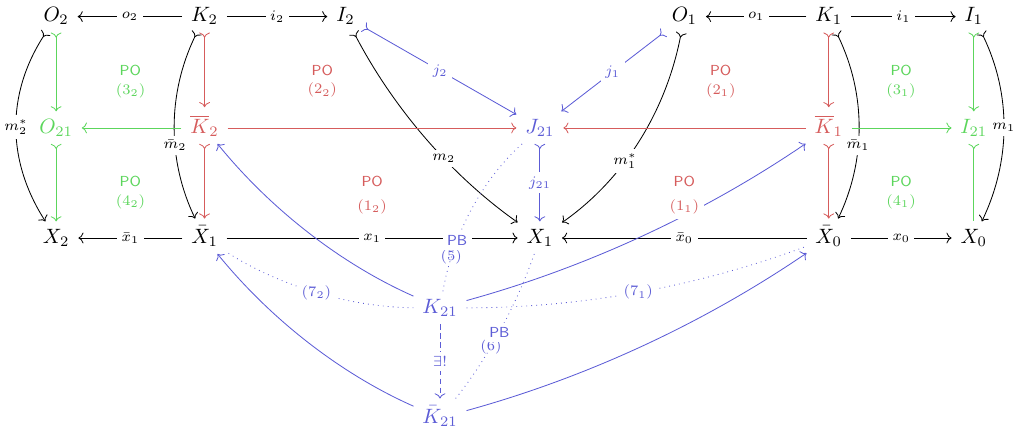}
\end{equation}
\begin{itemize}
\item By the \kl{universal property of pullbacks}, there exists a morphism $K_{21}\rightarrow \bar{K}_{21}$.
\item Since $(1_1)$ and $(1_2)$ are pushouts along $\cM$-morphisms and thus pullbacks, by \kl{pullback-pullback decomposition} both $(7_1)$ and $(7_2)$ are pullbacks. Thus by \kl{stability of $\cM$-morphisms under pullbacks}, $K_{21}\rightarrow \bar{K}_{21}$ is in $\cM$. Since moreover in a quasi-topos, \kl{pushouts along $\cM$-morphisms are stable under pullbacks}, we find\footnote{Note in particular that at this point, we are \emph{not} using the rm-van Kampen property, i.e., at least the ``synthesis'' construction would indeed be feasible via requiring $\bfC$ to be merely a quasi-topos and not an rm-adhesive category. However, we shall see that the ``analysis'' part of the proof relies upon precisely the rm-van Kampen property in a crucial step.} that $(7_1)$ and $(7_2)$ are pushouts.
\item Finally, by \kl{pushout composition}, we find that $(4_1)+(7_1)$ and $(4_2)+(7_2)$ are pushouts, whence in summary we have succeeded in exhibiting a DPO-type direct derivation of $X_0$ along the composite rule $(O_{21}\leftarrow K_{21}\rightarrow I_{21})$ along the admissible match $(I_{21}\rightarrowtail X_0)$.
\end{itemize}

\subsection{``Analysis'' part}

Suppose we were given a DPO-type direct derivation of an object $X_0\in \obj{\bfC}$ along a DPO-type composite rule as specified via the data of the diagram below:
\begin{equation}
\ti{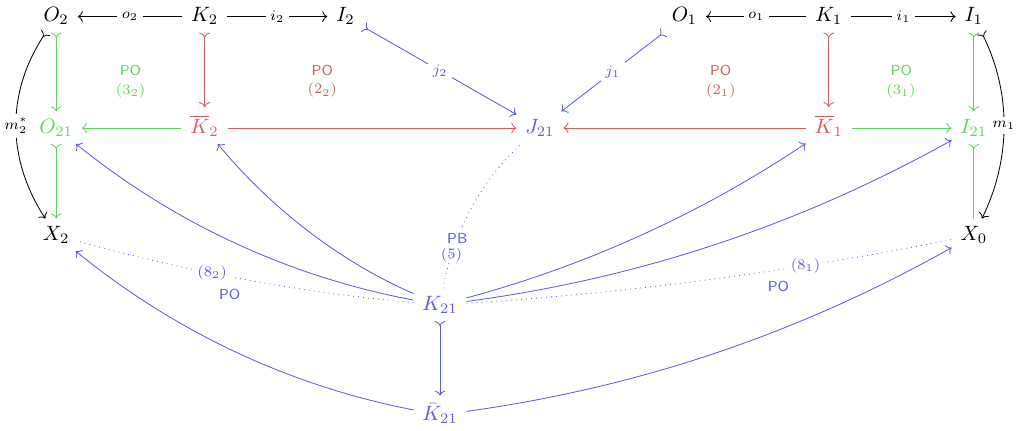}
\end{equation}

Construct the diagram below by taking pushouts to form the squares marked $(9_1)$, $(9_2)$ and $(10)$ (where we could have chosen w.l.o.g.\ to form square $(11)$ instead):
\begin{equation}
\ti{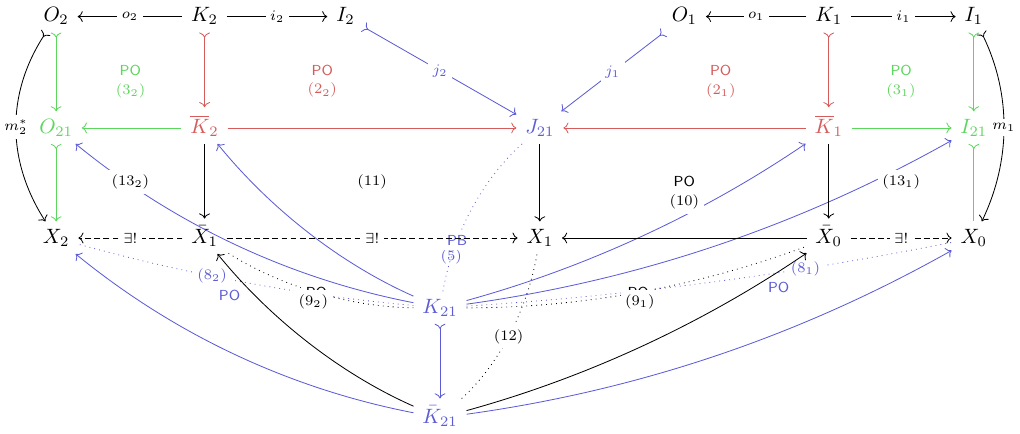}
\end{equation}
\begin{itemize}
\item By the \kl{universal property of pushouts}, there exist morphisms $\bar{X}_0\rightarrow X_0$, $\bar{X}_1\rightarrow X_1$ and $\bar{X}_1\rightarrow X_2$.
\item By \kl{pushout-pushout decomposition}, the squares marked $(11)$,  $(13_1)$ and $(13_2)$ are pushouts, and thus by \kl{stability of $\cM$-morphisms under pushouts}, all vertical morphisms of these squares are in $\cM$.
\item Since $\bfC$ is by assumption an \kl{rm-adhesive category}, by the rm-van Kampen property the square marked $(12)$ is found to be a pullback.
\item Finally, by \kl{pushout composition}, the squares $(3_1)+(13_1)$, $(2_1)+(10)$, $(3_2)+(11)$ and $(3_2)+(13_2)$ are pushouts, whence we have indeed exhibited a sequence of two DPO-type direct derivations.
\end{itemize}

\subsection{``Compatibility'' part}

Upon closer inspection of the steps followed in the ``synthesis'' and ``analysis'' parts of the proof, respectively, note first that both parts render precisely the same structure of commutative diagrams (where it is important to note that full compatibility is ensured via the use of the rm-van Kampen property in order to demonstrate that square $(12)$ is a pullback). Therefore, the set of pairs of DPO-admissible matches of rules into objects and the set of pairs of a DPO-admissible match of rules and a DPO-admissible match of the composite rules into objects are isomorphic precisely if we quotient each set by the relevant notion of universal isomorphisms (i.e., the former set by the isomorphisms arising from the \kl{universal property of multi-POCs}, and the latter set by isomorphisms arising from the \kl{universal property of multi-sums} as well as by the \kl{universal property of multi-POCs}, respectively).

\end{document}